\documentclass[final,authoryear,5p,twocolumn,11pt]{elsarticle}

\usepackage[parfill]{parskip}    
\usepackage{palatino,graphicx,latexsym,setspace,xspace,subfigure,epstopdf}
\usepackage{amsmath,amsthm,amssymb,wasysym,wrapfig}
\usepackage{soul,color, lineno}
\usepackage{nopageno}
\usepackage{lscape}
\pdfoutput=1

\newtheorem{theorem}{Theorem}
\newtheorem{lemma}[theorem]{Lemma}

\journal{arXiV}

\begin{document}

\begin{frontmatter}

\title{On the geometric structure of fMRI searchlight-based information maps}

\author{Shivakumar Viswanathan \fnref{label1}}

\author{Matthew Cieslak}

\author{Scott T. Grafton}

\fntext[label1]{Corresponding author: shiva@psych.ucsb.edu}
\address{Department of Psychological and Brain Sciences, \\
University of California, Santa Barbara, California, USA}

\begin{abstract}
Information mapping is a popular application of Multivoxel Pattern Analysis (MVPA) to fMRI. Information maps are constructed using the so called searchlight method, where the spherical multivoxel neighborhood of every voxel (i.e., a searchlight) in the brain is evaluated for the presence of task-relevant response patterns. Despite their widespread use, information maps present several challenges for interpretation. One such challenge has to do with inferring the size and shape of a multivoxel pattern from its signature on the information map. To address this issue, we formally examined the geometric basis of this mapping relationship. Based on geometric considerations, we show how and why small patterns (i.e., having smaller spatial extents) can produce a larger signature on the information map as compared to large patterns, independent of the size of the searchlight radius. Furthermore, we show that the number of informative searchlights over the brain increase as a function of searchlight radius, even in the complete absence of \emph{any} multivariate response patterns. These properties are unrelated to the statistical capabilities of the pattern-analysis algorithms used but are obligatory geometric properties arising from using the searchlight procedure.
\end{abstract}

\end{frontmatter}


\pagebreak
\section{Introduction}

In FMRI, a functional map is an important representation of how cognitive function is related to neuroanatomy. Such maps provide a topographic representation of the brain regions that are (and are not) systematically responsive to differing values of a cognitive variable. The size, shape, number and location of the ``blobs" (i.e., voxel-clusters meeting some statistical relevance criterion) on the functional maps are the basis for inferences about the neural substrates of the cognitive process. Given the importance of functional maps, there is a continuing need to scrutinize the sensitivity, precision and technical assumptions of the mapping procedure itself. The topic of the current technical note is the mapping procedure used to generate the widely used \emph{information maps} \citep{Kriegeskorte:2006bf}. 

The motivation for information mapping is the statistical concern that a region's responses to the cognitive variable under study might take a complex multivariate form. For example, a group of multiple voxels might conjointly respond in a task-relevant manner even though individual voxels may not detectably do so  \citep{Haxby:2001kl,Cox:2003vc,Haynes:2006vv,Norman:2006gy,Mur:2008cn,Tong2010,Wagner2010,Formisano:2012fk,Serences:2012fk}. Such distributed response patterns might be effectively undetectable with conventional univariate statistical tests restricted to individual voxel responses, but detectable with an explicit multivariate test for multivoxel response patterns, i.e., using some Multivoxel Pattern Analysis (MVPA) technique. To address this concern in the context of functional mapping,  \citet{Kriegeskorte:2006bf} proposed a simple procedure to enable sophisticated MVPA methods to be readily applied to detect and map brain regions that contain information about the experimental conditions, irrespective of whether the informative responses are univariate or multivariate. 

In the proposed procedure, the unit of evaluation is not the single voxel but a ``searchlight" -- the \emph{group} of voxels contained in a spherical neighborhood of radius $r$ around a single voxel. The searchlight statistic is a measure of whether the conjoint responses of this group of voxels contain information about the experimental conditions being tested. Based on these abstractions, an information map is generated as follows: the searchlight statistic is evaluated for searchlights centered at \emph{every} voxel in the brain; and the statistic's value for each searchlight is mapped to the central voxel of that searchlight. The resulting topographic representation generated by the searchlight-procedure has been referred to as the information map. Such searchlight-based information maps are now routinely reported in studies that employ MVPA methods \cite[for example,][]{Haynes:2007dy,Soon:2008bl,Johnson:2009kx,Poldrack:2009fm,Chadwick:2010uq,Oosterhof:2010jq,Nestor:2011pf,Alink:2011sf,Golomb:2011ph,Peelen:2011jt,Stokes:2011ly,Woolgar:2011kx,Morgan:2011ys,Oosterhof:2012vb,Connolly:2012gl,Kaplan:2012gd}. 

Notwithstanding their popularity, interpreting a searchlight-based information map presents a variety of challenges \citep{Kriegeskorte:2006bf,Poldrack:2009fm,Pereira:2011bl,Jimura:2012cs}. One such challenge is posed by the topographic ambiguity of the information map. Recall that a searchlight statistic computed on the responses of an entire multivoxel searchlight is mapped to a single voxel on the information map, namely, that searchlight's central voxel. This mapping protocol is applied to searchlights across the brain irrespective of the number or the spatial locations of the information-carrying voxels \emph{within} each searchlight. Consequently, the spatial position of an informative voxel on the information maps is a coarse index to the actual location of the informative ``pattern" within that voxel's searchlight. Furthermore, since a searchlight has a unique central voxel $v$, an informative voxel on the information map is not indicative of the actual number of voxels constituting the informative pattern within that voxel's searchlight neighborhood. Given these properties of the information map, we asked: what, if anything, can be reliably inferred about the size and shape of a multivoxel pattern from its corresponding signature on the information map?

Previous studies have treated this question as a qualitative concern requiring cautious interpretation. Nonetheless, here we show that information maps are in fact subject to several crisply quantitative geometric constraints that strongly govern how such maps can be interpreted.

Our analytical results are based on a simple geometric intuition. Since a multivoxel searchlight is defined at every voxel across the brain, searchlights centered at different voxels systematically overlap each other, i.e., have voxels in common. Using overlapping searchlights is crucial to obtain a continuous topographic coverage especially when the locations and spatial extents of voxel-neighborhoods that are task-responsive are unknown a priori. We observed that due to these overlaps, multiple searchlights would be deemed informative merely by virtue of sharing the \underline{same} task-relevant multivoxel response patterns. Thus we reasoned that the size and shape of a multivoxel group $G$'s signature on the information map should be defined by exactly those voxels which have searchlight-neighborhoods that contain $G$. Using this observation and simple geometric reasoning, we formally deduce some key properties of the relationship between an informative pattern and its corresponding signature on the information map. 

Based on our formal analysis, we prove here that, for any searchlight radius, a single task-responsive voxel produces a larger signature on the information map as compared to a distributed multivoxel response pattern. Furthermore, the number of informative searchlights over the brain can increase as a function of searchlight radius, without necessarily revealing any new information and even in the complete absence of \emph{any} multivariate response patterns. Importantly, these properties are largely independent of the type of machine-learning algorithm or the testing protocol used to compute the searchlight statistic.

\section{Model \label{sec:overlap}}

\subsection{Definition: The searchlight decomposition}

The basis of the searchlight analysis is the geometric structure of the voxel-space in which the brain images are defined. The voxel-space $\mathcal{V}$ is defined here as the set of all voxels $V$ augmented with a geometric structure defining the relative spatial position of the voxels in $V$, and a distance measure between these voxels. For analytical convenience, we treat the voxel-space as being \emph{uniform} and \emph{connected} as described below.

A $d$-dimensional voxel-space $\mathcal{V}$ is deemed to be \emph{uniform} if every voxel has a neighboring voxel in all $d$ principal directions. Additionally, we assume that the voxel-space $\mathcal{V}$ is \emph{connected}. Specifically, there is a path connecting every pair of voxels $v_i$ and $v_j$ in $\mathcal{V}$ with a path defined here to be an ordered sequence of voxels $\langle v_i, \cdots, v_{k}, v_{k+1},....v_j \rangle$ where voxel $v_{k+1}$ is a neighbor of $v_{k}$ along one of the $d$ principal directions. These simplifying assumptions are intended to emphasize the general geometric principles entailed by the searchlight method while deliberately ignoring the special cases associated with (i) the boundaries of $\mathcal{V}$ where a searchlight may be truncated; and  (ii) distinctions between gray-matter and white-matter voxels and any masking of the latter from the searchlights. Although we refer to searchlights as being volumes in a voxel-space having dimensionality $d=3$, the properties derived here are agnostic to the specific value of $d$ and apply to surfaces ($d=2$) where the searchlights are discs \cite[as in,][]{Oosterhof:2011ad,Chen:2011vl}.

The key abstraction defined by the searchlight method is a decomposition of $\mathcal{V}$ into subsets of voxels based on a geometric criterion. Given a voxel space $\mathcal{V}$, we define a searchlight voxel-decomposition using the following indexing function 
\begin{equation}
S: \mathcal{V} \times \mathbb{R} \rightarrow \mathcal{P}(V)
\end{equation}
where $\mathcal{P}(\mathcal{V})$ is the powerset of $\mathcal{V}$, namely, the set of all subsets of $\mathcal{V}$. This indexing function $S$ takes two inputs -- the identity of a voxel $v$ in the voxel-space $\mathcal{V}$, and a real-value $r \in \mathbb{R}$ specifying the searchlight's radius. The searchlight indexing function uses these parameters in conjunction with the geometric structure of $\mathcal{V}$ to extract and output a set of voxels $S(r,v) \in \mathcal{P}(V)$. A voxel $v' \in \mathcal{V}$ is a member of $S(r, v)$ if and only if the distance between $v'$ and $v$ is less than or equal to $r$. For convenience, we henceforth write $S_r(v)$ to denote the searchlight $S(r,v)$. The resulting \emph{searchlight voxel-decomposition} of $\mathcal{V}$ for a given radius $r$ is defined as 
\begin{equation}
\mathcal{S}_r(\mathcal{V}) = \{S_r(v) \; | \; \text{for all} \; v \in \mathcal{V} \}
\label{eq:decomp}
\end{equation}

For clarity, we restrict our usage of the term ``searchlight" to the cases when the value of the radius of a searchlight $r$ is such that  each $S_r(V)$ is a \emph{multivoxel} entity that is not identical with $\mathcal{V}$, that is, $1 < |S_r(v)| \ll |\mathcal{V}|$,  for any $S_r(v) \in \mathcal{S}_r(\mathcal{V})$. We refer to the univariate case where $|S_r(v)|=1$ as the univoxel decomposition.

\begin{figure*}[!tb]
\begin{center}
\includegraphics[height=2.5in]{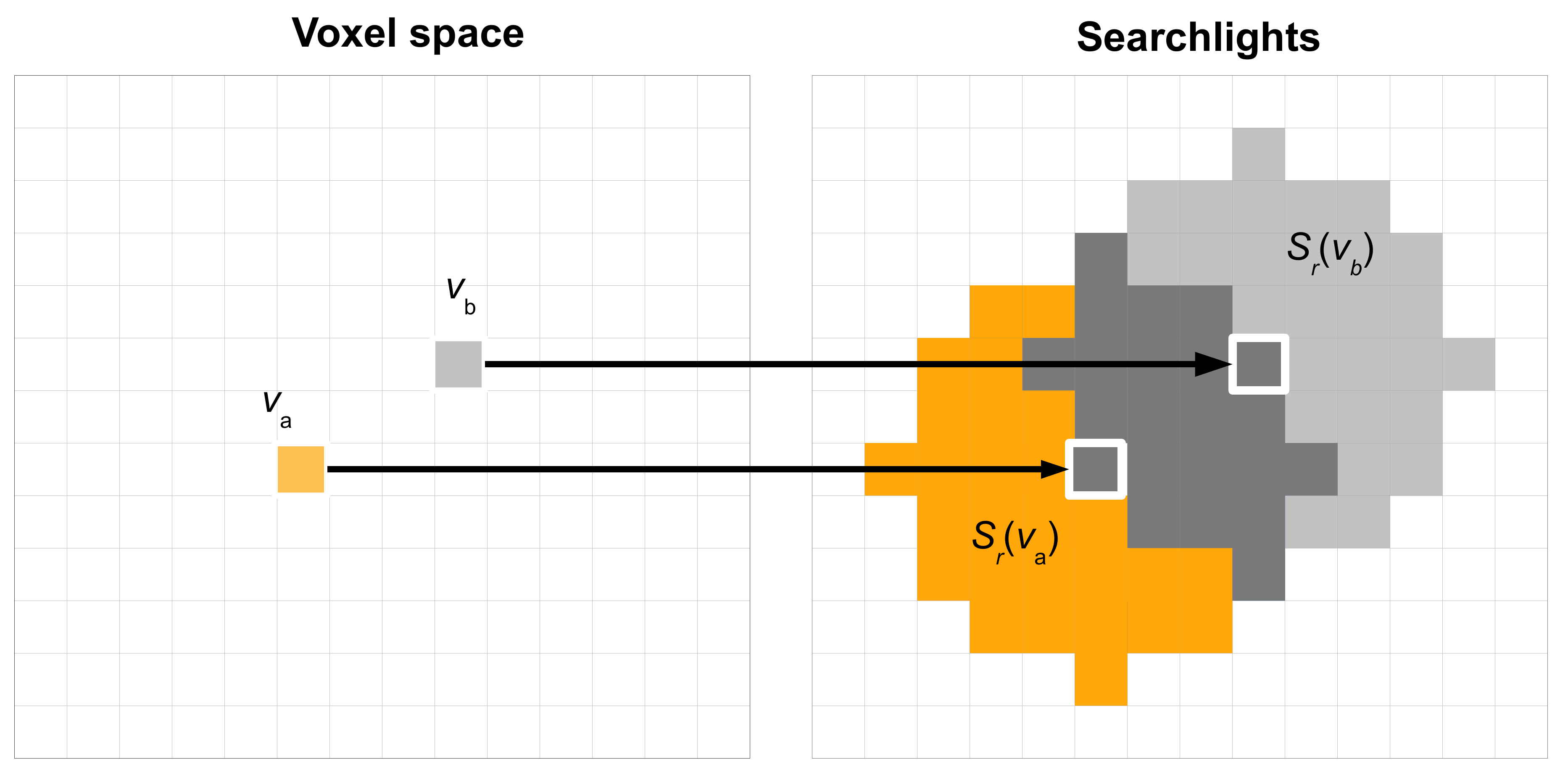}
\caption{{\bf Cartoon of searchlight indexing scheme}: Left panel shows two example voxels $v_a$ and $v_b$ that are mapped to searchlights $S_r(v_a)$ (orange) and $S_r(v_b)$ (light gray) respectively, having some radius $r$, as shown in the right panel. Two searchlights can overlap to varying degrees depending on the distance between their center voxels and the radius. Here the overlap is indicated in dark-gray. The searchlight statistic computed on each searchlight is mapped back to the corresponding central voxels to generate the information map (see text).}
\label{fig:voxelspace}
\end{center}
\end{figure*}

A schematic of the searchlight indexing scheme is shown in Figure \ref{fig:voxelspace}.

\subsection{Definition: Informativeness function}

The searchlight statistic is a measure of whether the voxels in the searchlight, as a unit, exhibit differences in their conjoint responses to the experimental conditions. More generally, it is a measure of whether the searchlight contains information about the experimental condition, i.e., whether the searchlight is informative. As with the radius of the searchlight, the specific statistical procedure used to compute the searchlight statistic is a discretionary choice made by the researcher \citep[for example, see][]{Pereira:2011bl}.

To describe the searchlight statistic in a procedure-independent manner, we use a binary indicator function, which we refer to as the \emph{informativeness} function, $I: \mathcal{P}(V) \rightarrow \{0, 1\}$. Given a subset of voxels $G \in \mathcal{P}(V)$, the function $I$ returns a value of $1$ if the responses of $G$ are deemed to be informative; or $0$ if they are not, based on some appropriately specified statistical criterion. 

Evaluating the informativeness function on the responses corresponding to each searchlight $S_r(v)$ in 
$\mathcal{S}_r(\mathcal{V})$ defines the overall information set for a particular radius
\begin{equation}
\mathbf{I}(\mathcal{S}_r(\mathcal{V})) = \{I(S_r(v)) \; | \; \text{for all} \; S_r(v) \; \text{in} \; \mathcal{S}_r(\mathcal{V}) \} 
\end{equation}

The \emph{information map} is the object obtained when the information set defined above is augmented with the geometric structure of the voxel-space $\mathcal{V}$ by mapping the informativeness value of each searchlight 
$I (S_r(v))$ to its corresponding central voxel $v$. 

The performance measure of interest here is the total number of informative searchlights for a particular searchlight decomposition
\begin{equation}
F_r = \sum_{i = 1}^{|\mathcal{V}|} I(S_r(v_i))
\label{eq:totalnum}
\end{equation}

\subsection{Linking assumptions}

Two simple properties link the structure of the searchlight decomposition $\mathcal{S}_r(\mathcal{V})$ to the structure of the information set $\mathbf{I}(\mathcal{S}_r(\mathcal{V}))$. 

The first property is that, by virtue of the regularity of their shape and relative positioning, searchlights in 
$\mathcal{S}_r(\mathcal{V})$ can overlap. Consequently, the same voxels in $\mathcal{V}$ can be included or sampled by multiple searchlights. The second property is that since a sphere is an arbitrarily chosen and regular shape, it is unlikely that every voxel is necessarily task-relevant in every informative searchlight. Consequently, the informativeness of the responses in some particular searchlight volume $S_r(v)$ can be alternatively and accurately interpreted as indicating that \emph{some} group of voxels in that searchlight volume exhibits task-dependent responses. These two properties can be combined as follows. Let $G$ be a group of task-relevant voxels in a searchlight $S_r(v)$, where $G \subseteq S_r(v)$. Since searchlights share voxels, if some other searchlight $S_r(v')$ also contains $G$, that is $G \subset S_r(v')$, then it implies that $S_r(v')$ should also contain task-relevant information as it includes the task-relevant voxels $G$. 

Based on this observation, we make two linking assumptions about the behavior of the procedures used to compute the searchlight statistic and hence the informativeness function $I$. The first is that we restrict the focus of our analysis to the common multivariate procedure that does \underline{not} include geometric information about the relative spatial positions of the voxels in a searchlight while computing that searchlight's informativeness. The second is the \emph{Superset informativeness} (SIN) assumption which postulates that:

\textbf{Superset informativeness assumption}: \emph{If a group of voxels $G$  is informative then every searchlight that contains $G$ is also informative.}

That is, according to the SIN assumption, if $I(G) = 1$ then $I(S_r(v))=1$ for all $v \in V$ where $|G|>0$ and $G \subseteq S_r(v)$. Unless otherwise stated, we will overload the symbol $I$ to denote an informativeness function that explicitly satisfies these two model requirements. 

Although the SIN assumption is based is on a sound deduction, the empirical requirement that it poses may not necessarily be satisfied in practice. Specifically, even if it is known that $I(G)=1$, the statistical procedure used to evaluate informativeness might fail to detect that a searchlight $S_r(v)$ is informative even if $G \subset S_r(v)$. Such a Type II error (i.e., failing to reject a false null hypothesis that $H_0: I(S_r(v)) = 0$) might occur for any of a variety of reasons, for example, the use of an inappropriate machine-learning algorithm \citep{Pereira:2011bl}, insufficient power due to a limited number of samples, and so on. In this regard, the SIN assumption treats the multivoxel pattern analysis techniques as being more sensitive and reliable than might actually be the case in practice.  That is, the SIN assumption allows us to establish the information map's properties in the best-case independent of the performance idiosyncrasies of the specific multivariate method being used.

\section{Analytical results}\label{sec:analytical}

Our focus of the current section is to establish how the structure of the sampling bias arises from the searchlight decomposition. We first prove that due to the geometric regularities of a searchlight decomposition, single-voxels and multivoxel groups are sampled with different frequencies, i.e., included in a different number of searchlights. Specifically, single voxels are included in more searchlights than multivoxel groups. This sampling difference is independent of the searchlight radius. We then extend these results to prove that the frequency with which voxel-groups are sampled increases with the radius of the searchlights, irrespective of the number of voxels in the group.  Finally, we prove that the information map mirrors these sampling  biases in an optimistic manner, i.e., in a manner that is not necessarily warranted by the data.

\subsection{Single-voxels and multivoxel-groups are sampled with different frequencies}\label{sec:sampling}

The regularity in the shape of the searchlights and their relative positions the voxel-space define a systematic relationship between each voxel $v \in \mathcal{V}$ and the searchlights in $\mathcal{S}_r(\mathcal{V})$ that contain that voxel $v$. Firstly, if a voxel $v_a$ is a member of the searchlight $S_r(v_b)$, then by symmetry, the voxel $v_b$ is a member of the searchlight $S_r(v_a)$ (Lemma \ref{prop:symmetry}). Secondly, two distinct voxels $v_a$ and $v_b$ are not simultaneously included in every searchlight that contains either of these voxels (Lemma \ref{prop:twovoxel}).

\vspace{10pt}
\begin{lemma}\label{prop:symmetry}
If a voxel $v_b$ is a member of $S_r(v_a)$ then the voxel $v_a$ is a member of $S_r(v_b)$, where $v_a, v_b \in \mathcal{V}$ and $S_r(v_a), S_r(v_b) \in \mathcal{S}_r(\mathcal{V}$.
\end{lemma}

\begin{proof} 
Consider a searchlight $S_r(v_a)$ centered at voxel $v_a$. Since a searchlight is defined at every voxel in $\mathcal{V}$ (Equation \ref{eq:decomp}), it follows that there is a searchlight defined at every voxel in $S_r(v_a)$. By definition, a voxel $v_b \in \mathcal{V}$ is a member of $S_r(v_a)$ if and only if the distance between $v_a$ and $v_b$ is less than or equal to the radius $r$. Since there is a searchlight $S_r(v_b)$ centered at $v_b \in S_r(v_a)$, and the distance between $v_a$ and $v_b$ is less than or equal to $r$, it follows that $v_a$ is a member of searchlight $S_r(v_b)$. Therefore, if $v_b$ is a member of $S_r(v_a)$ then $v_a$ is a member of $S_r(v_b)$. 
\end{proof}
     
 \vspace{10pt}    
\begin{lemma} \label{prop:twovoxel} For any two non-identical voxels $v_a$ and $v_b$, where $v_a, v_b \in \mathcal{V}$ and $S_r(v_a) \neq S_r(v_b)$, there necessarily exists a searchlight that contains $v_a$ but not $v_b$ and a different searchlight that contains $v_b$ but not $v_a$. 
\end{lemma}

\begin{proof} 
This claim can be proved in two steps based on the distance between $v_a$ and $v_b$.

First consider the case where the distance between $v_a$ and $v_b$ is greater than $2r$, that is, the diameter of a searchlight. By definition, a searchlight contains voxels that have a distance less than or equal to $r$ from that searchlight's central voxel. Due to the spherical shape of the searchlight, the maximum distance between any two voxels in a searchlight is equal to $2r$. If the distance between $v_a$ and $v_b$ is greater than $2r$, there does not exist any searchlight of radius $r$ that contains both $v_a$ and $v_b$ as members. Thus, it follows that there exists some searchlight that contains $v_a$ but not $v_b$; and some other searchlight that contains $v_b$ but not $v_a$.

Now consider the second case where the distance between $v_a$ and $v_b$ is less than or equal to $2r$. Since the distance between these two voxels is less than the maximum distance between some two voxels in a searchlight, in a uniform voxel-space there necessarily exists some searchlight $S_r(v)$ that contains both $v_a$ and $v_b$ as members. Contrary to the proposition, let us assume that both $v_a$ and $v_b$ are contained in every searchlight that contains either $v_a$ or $v_b$. That is, if a searchlight $S_r(v)$ contains $v_a$, then it necessarily contains $v_b$, and vice versa.  Recall that, from Lemma 1, a voxel $v_a$ is contained in every searchlight $S_r(v)$ where $v \in S_r(v_a)$. Now, based on the contradictory assumption, it implies that $v_b$ is also contained in every such searchlight $S_r(v)$ where $v \in S_r(v_a)$. By the same reasoning, $v_a$ should be contained in every searchlight $S_r(v)$ where $v \in S_r(v_b)$. If these conditions hold true, then it implies that every voxel in $S_r(v_a)$ is also contained in $S_r(v_b)$; and every voxel in $S_r(v_b)$ is also contained in $S_r(v_a)$. If this the case, then the searchlights $S(v_a)$ and $S(v_b)$ are identical as they contain exactly the same voxels. This relationship, however, contradicts the requirement that $S_r(v_a) \neq S_r(v_b)$. Thus, the assumption that $v_a$ and $v_b$ are both contained in every searchlight that contains either $v_a$ or $v_b$ cannot be true. 

 Therefore, there necessarily exists a searchlight that contains $v_a$ that does not contain $v_b$, and some other searchlight that contains $v_b$ but not $v_a$. 
\end{proof}

Armed with the properties described by Lemmas \ref{prop:symmetry} and \ref{prop:twovoxel}, we can now numerically estimate the number of searchlights that include a given individual voxel.
  
\vspace{10pt}
\begin{theorem}\label{cor:corsym}
A voxel $v$ is contained in exactly $N_r(v)$ different searchlights, where $N_r(v)$ is the number of voxels contained in the searchlight $S_r(v)$.
\end{theorem}
\begin{proof}
From Lemma \ref{prop:symmetry}, a voxel $v$ is contained in each searchlight $S_r(v')$, if and only if $v'$ is a voxel in $S_r(v)$. Let $N_r(v)$ be the number of voxels in $S_r(v)$. Therefore, $v$ is present in each of these $N_r(v)$ searchlights. 
\end{proof}
 
For simplicity, we treat $N_r(v)$ as being the same for every searchlight, and write $N_r$ to indicate the canonical number of voxels contained in a spherical volume of radius $r$, for a given resolution of the voxel-space. 

\begin{figure}[tb]
\begin{center}
\includegraphics[width=0.5 \textwidth]{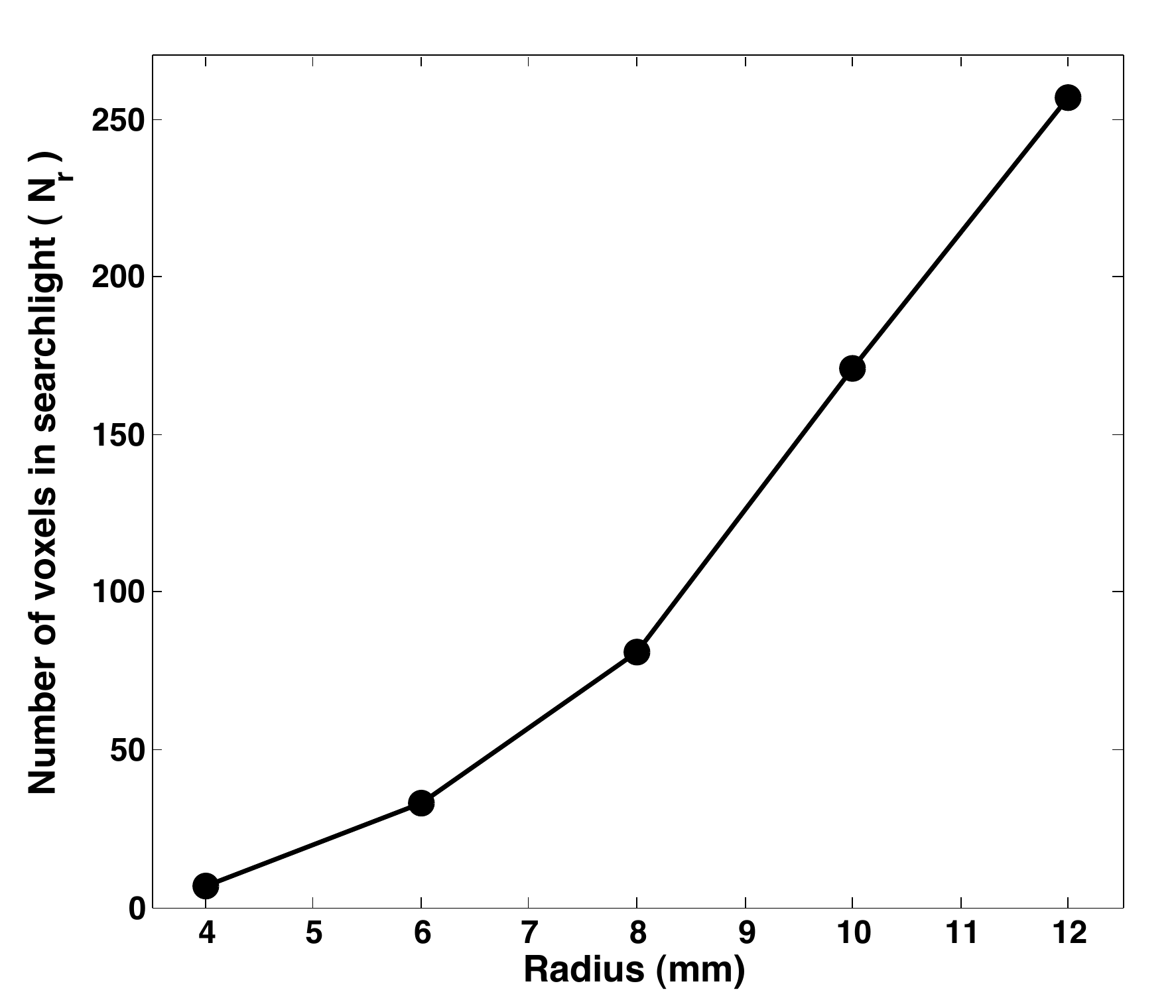}
\caption{Number of voxels in a searchlight volume ($N_r$) as a function of searchlight radius ($r$) in a 3$\times$3$\times$3 mm$^3$ voxel-space.}
\label{fig:numvoxel}
\end{center}
\end{figure}

From Theorem \ref{cor:corsym}, we see that the radius, a parameter chosen by the researcher, directly specifies how often information in a particular voxel is sampled by multiple searchlights. For voxels of size 3mm $\times$ 3mm $\times$ 3mm, the number of voxels contained in searchlights of different radii are shown in Figure \ref{fig:numvoxel}. As can be seen, the number of voxels in a searchlight, that is $N_r$, grows rapidly with the radius of the searchlight $r$, and consequently so do the number of searchlights that include a particular voxel. 

Since searchlights are intended to identify multivoxel response patterns, we extend the single-voxel property in Theorem \ref{cor:corsym} to quantify the membership of a group of multiple voxels placing no constraint on the relative spatial locations of the voxels in the group. 

\vspace{10pt}
\begin{theorem} \label{cor:haystack}
A group of voxels $G$ containing more than one voxel is contained in strictly less than $N_r$ searchlights.
\end{theorem}

\begin{proof} 
A voxel is contained in $N_r$ searchlights, from Theorem \ref{cor:corsym}. Consequently, every voxel in $G$ is each contained in $N_r$ searchlights. From Lemma \ref{prop:twovoxel}, for any two voxels $v_a$ and $v_b$, there is necessarily a searchlight that contains $v_a$ and not $v_b$, and vice versa. Therefore, of the $N_r$ searchlights containing $v_a$, there necessarily exists at least one searchlight that contains $v_a$ but not $v_b$. Thus, the number of searchlights that simultaneously contain both $v_a$ and $v_b$ must be less than $N_r$. Since $G$ contains multiple voxels, any pair of voxels in $G$ must be simultaneously contained in less than $N_r$ searchlights. Therefore, all the voxels in $G$ cannot be simultaneously contained in $N_r$ searchlights, and $G$ must be contained in strictly less than $N_r$ searchlights. 
\end{proof}

\subsection{The sampling frequency of voxel(s) increases with searchlight radius}\label{sec:scaling}

Although single-voxels and multivoxel groups are included in different numbers of searchlights for any radius $r$, we now show that the absolute number of searchlights that include either a single-voxel or a multivoxel group increases with the radius of the searchlight. 


\vspace{10pt}
\begin{lemma} \label{prop:inclusion}A searchlight of radius $r_a$ is fully contained in more than one searchlight of radius $r_b$, where $r_b > r_a$ and $S_{r_a}(v) \subset S_{r_b}(v)$ for all $v \in \mathcal{V}$. 
\end{lemma}

\begin{proof}  
Consider two searchlights centered at the same voxel $v$ -- one that has a radius $r_a$, and the other having radius $r_b$. By definition, since $S_{r_a}(v) \subset S_{r_b}(v)$, all the voxels in $S_{r_a}(v)$ are members of $S_{r_b}(v)$, and there exists at least one voxel in $S_{r_b}(v)$ that is not in $S_{r_a}(v)$.

Now, consider the searchlight $S_z(v)$, having radius $z = r_b - r_a$.  Due to the spherical shape of searchlights, the maximum distance between a voxel $v''$ in $S_z(v)$ and some voxel $v'$ in $S_{r_a}(v)$ is equal to $r_a + z = r_b$. Therefore, all other voxels in $S_{r_a}(v)$ must have distances less than or equal to $r_b$.

Since the distance between these two maximally distant voxels $v'$ and $v''$ is equal to $r_b$, the voxel $v'$ must be contained in a searchlight of radius $r_b$ that is centered at $v''$, namely, $S_{r_b}(v'')$. Since all other voxels in $S_{r_a}(v)$ have a distance less than or equal to $r_b$ from $v''$, it follows that every voxel in $S_{r_a}(v)$ is also contained in the searchlight $S_{r_b}(v'')$. Thus every voxel in $S_{r_a}(v)$ is contained in at least two searchlights having radius $r_b$, namely, $S_{r_b}(v)$ and $S_{r_b}(v'')$. Therefore, a searchlight $S_{r_a}(v)$ is contained in more than one searchlight of radius $r_b$, where $r_b > r_a$ and $S_{r_a}(v) \subset S_{r_b}(v)$. 
\end{proof}

Using Lemma \ref{prop:inclusion}, we can now prove a general scaling property. Irrespective of the size of a voxel group, the frequency with which it is sampled by different searchlights increases with the radius of the searchlight - a property that we prove next. 

\vspace{10pt}
\begin{theorem} \label{cor:scaling}A group of voxels $G$ is contained in more searchlights of radius $r_b$ than searchlights of radius $r_a$, where $G\subseteq S_{r_a}$, $r_b > r_a$ and $S_{r_a}(v) \subset S_{r_b}(v)$, for all $v \in \mathcal{V}$.
\end{theorem}

\begin{proof} Let $K_{r_{a}}$ and $K_{r_{b}}$ be the number of searchlights of radius $r_{a}$ and $r_{b}$ that contain $G$.

Since $S_{r_a}(v) \subset S_{r_b}(v)$ for every $v \in \mathcal{V}$, it follows, by transitivity, that if $G \subseteq S_{r_a}(v_{i})$ for some voxel $v_{i} \in \mathcal{V}$, then $G \subseteq S_{r_b}(v_{i})$. Therefore, the number of searchlights of radius $r_{b}$ that contain $G$ cannot be strictly less than that for $r_{a}$, that is, $K_{r_b}\nless K_{r_a}$.

By the transitivity of the subset relation, if $G \subseteq S_{r_a}(v_i)$ and $S_{r_a}(v_i) \subset S_{r_b}(v_j)$ for some $v_i,v_j \in \mathcal{V}$, then it follows that $G \subset S_{r_b}(v_j)$. From Lemma \ref{prop:inclusion}, a searchlight of radius $r_a$ is contained in multiple searchlights of radius $r_b$ where $r_b > r_a$ and $S_{r_a}(v) \subset S_{r_b}(v)$ (for all $v \in \mathcal{V}$). Since there is more than one searchlight of radius $r_b$ containing $S_{r_a}(v_{i})$, for every searchlight for which $G \subseteq S_{r_a}(v_{i})$ holds true, it implies that $K_{r_b} \geq K_{r_a}$.

From Theorems \ref{cor:corsym} and \ref{cor:haystack}, the number of searchlights of radius $r_a$ that can contain $G$ is less than or equal to $N_{r_a}$. Consequently, in a uniform and connected voxel-space, it follows that there exist two adjacent voxels $v_i$ and $v_j$ in $\mathcal{V}$ such that $G$ is a subset of $S_{r_a}(v_i)$ but is not a subset of $S_{r_a}(v_j)$. From Lemma \ref{prop:inclusion}, searchlights of radius $r_b$ centered at voxels within $r_b - r_a$ from $v_i$ fully contain all voxels in $S_{r_a}(v_i)$. Since $S_{r_a}(v) \subset S_{r_b}(v)$, it implies that the distance of $v_j$ to $v_i$ is less than or equal to $r_b - r_a$. Therefore, $S_{r_a}(v_i) \subset S_{r_b}(v_j)$ and consequently $G \subset S_{r_b}(v_j)$. Since $G \not\subset S_{r_a}(v_j)$ and $G \subset S_{r_b}(v_j)$, it implies that there exists at least one voxel at which a searchlight of radius $r_b$ contains $G$, but where a searchlight of radius $r_a$ does not contain $G$. Consequently, $K_{r_b}$ must be strictly greater than $K_{r_a}$, that is, the group of voxels $G$ is contained in more searchlights of radius $r_b$ than $r_a$. 
\end{proof}

Theorem \ref{cor:scaling} above establishes that the number of searchlights that include either a voxel or group of voxels increases monotonically with the radius of the searchlight. How then does this scaling of the sampling bias influence the properties of the information map?

\subsection{An optimistic bias in the information map \label{sec:smearing}}

Recall that $F_r$ (Equation \ref{eq:totalnum}) is an index of the sensitivity of the searchlight method in detecting multivoxel response patterns, and is equal to the total number of informative searchlights with a particular search decomposition. We now prove that as a direct consequence of how the sampling bias scales with the searchlight radius, the value of $F_r$ also increases strictly monotonically with increasing searchlight radius.

\vspace{10pt}
\begin{theorem}\label{thm:monotone}
For two searchlight radii, $r_a$ and $r_b$, where $r_b > r_a$ and $S_{r_a}(v) \subset S_{r_b}(v)$ for every $v \in \mathcal{V}$, if $0<F_{r_a}<V$ then $F_{r_b} > F_{r_a}$.
\end{theorem}

\begin{proof}
Since $S_{r_a}(v) \subset S_{r_b}(v)$ for every $v \in \mathcal{V}$, by the SIN assumption, it follows that if $I(S_{r_a}(v))=1$ then $I(S_{r_b}(v))=1$, for any voxel $v\in\mathcal{V}$. Therefore, the number of informative searchlights of radius $r_{b}$ cannot be strictly less than that for $r_{a}$, that is, $F_{r_b} \nless F_{r_a}$, for any value of $F_{r_a}$.

From Lemma \ref{prop:inclusion}, a searchlight of radius $r_a$ is contained in multiple searchlights of radius $r_b$ where $r_b > r_a$ and $S_{r_a}(v) \subset S_{r_b}(v)$ (for all $v \in \mathcal{V}$). For every searchlight for which $I(S_{r_a}(v))=1$, there is more than one searchlight of radius $r_b$ containing $S_{r_a}(v)$. Since each informative searchlight of radius $r_{a}$ is a subset of multiple searchlights of radius $r_{b}$, by the SIN assumption, it implies that $F_{r_b} \geq F_{r_a}$.

Let $0<F_{r_a}<V$. Since $F_{r_{a}} < V$, there necessarily exist two adjacent voxels $v_i$ and $v_j$ such that $I(S_{r_a}(v_i))=1$ and $I(S_{r_a}(v_j))=0$. By the same logic used to prove Theorem \ref{cor:scaling}, searchlights of radius $r_b$ centered at voxels within $r_b - r_a$ from $v_i$ fully contain all voxels in $S_{r_a}(v_i)$. Consequently, by the SIN assumption, $I(S_{r_b}(v_j)) =1$. This implies that a searchlight centered at voxel $v_j$ is informative if it has a radius $r_b$ but not if it has a radius $r_a$. Therefore $F_{r_b} > F_{r_a}$. 
\end{proof}
    
What does Theorem \ref{thm:monotone} have to do with optimism?  The monotonic increases in the number of informative searchlights is due to increases in the sampling bias, which in turn is due to the use of a multivoxel searchlight. Specifically, it is possible to obtain an increased ``sensitivity" of the information map simply by increasing the radius of the multivoxel searchlights, with no reference to the statistical properties of the voxel-responses, i.e., whether they in fact exhibit multivariate response differences.

\section{An illustration}\label{sec:simulation}

In this section, we present simulations to provide a concrete intuition for the analytical results above, and their implications. For ease of demonstration, the voxel-space $\mathcal{V}$ for all simulations consisted of a single axial slice having two principal directions. All the voxels in this voxel-space were populated with simulated response information from two fictitious experimental conditions $A$ and $B$.  These simulated data were subjected to the searchlight-procedure to produce information maps. The radius $r$ of the searchlights used for the searchlight decomposition $\mathcal{S}_r(\mathcal{V})$ was varied systematically to produce a corresponding information map for each radius value. The radius took the values: $4$ mm, $6$ mm, $8$ mm, $10$ mm and $12$ mm, corresponding to searchlights containing $5$ voxels, $13$ voxels, $21$ voxels, $37$ voxels and $49$ voxels respectively. 

The simulated response-data differed in the number and relative spatial location of the voxels that were responsive to the experimental conditions. In the first of these simulations discussed next, a single voxel contained task-relevant information while all the remaining voxels did not. 

\begin{figure*}[!tbp]
\begin{center}
\subfigure[]{\includegraphics[height=2.3in]{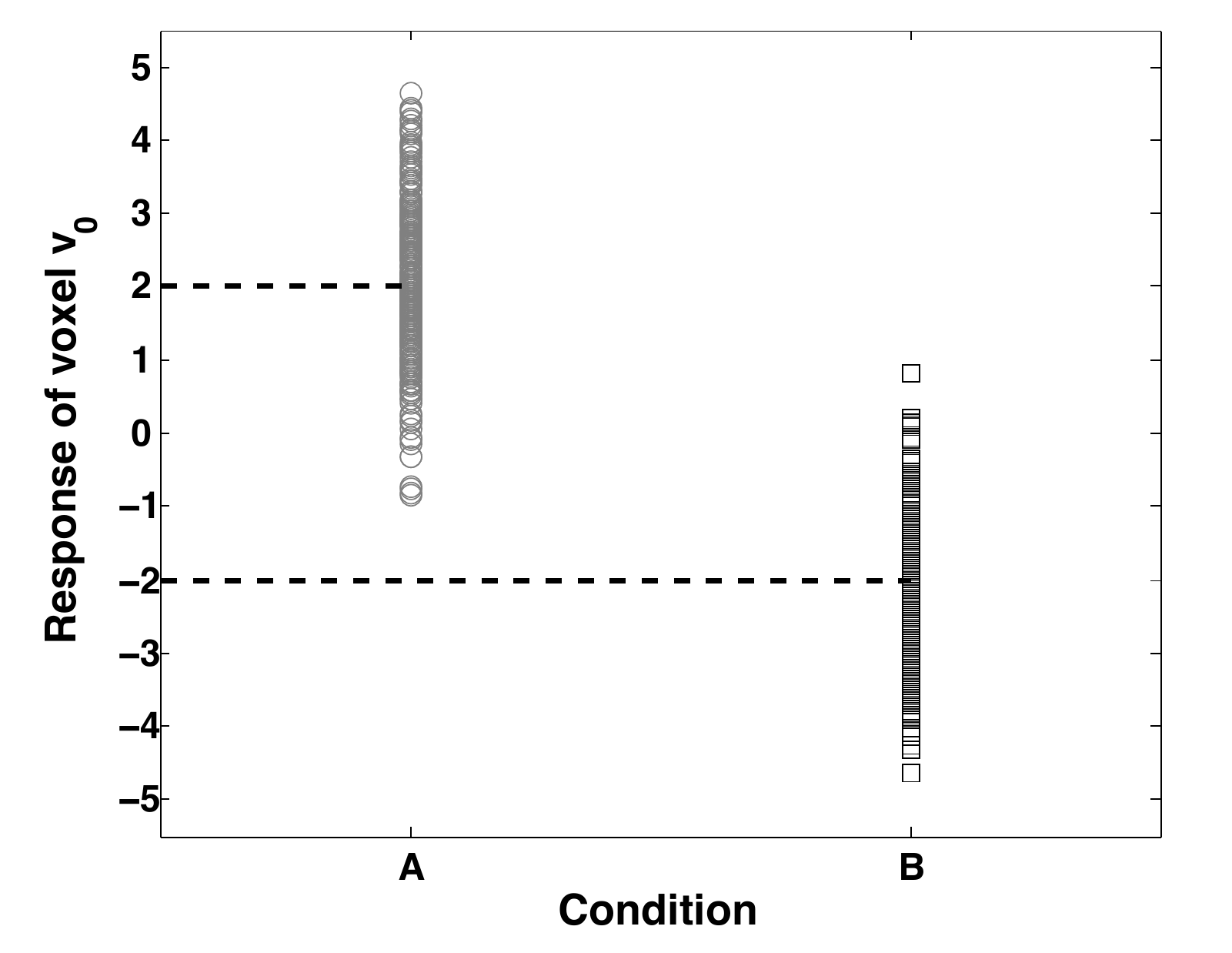}}
\hspace{0.5in}
\subfigure[]{\includegraphics[height=2.3in]{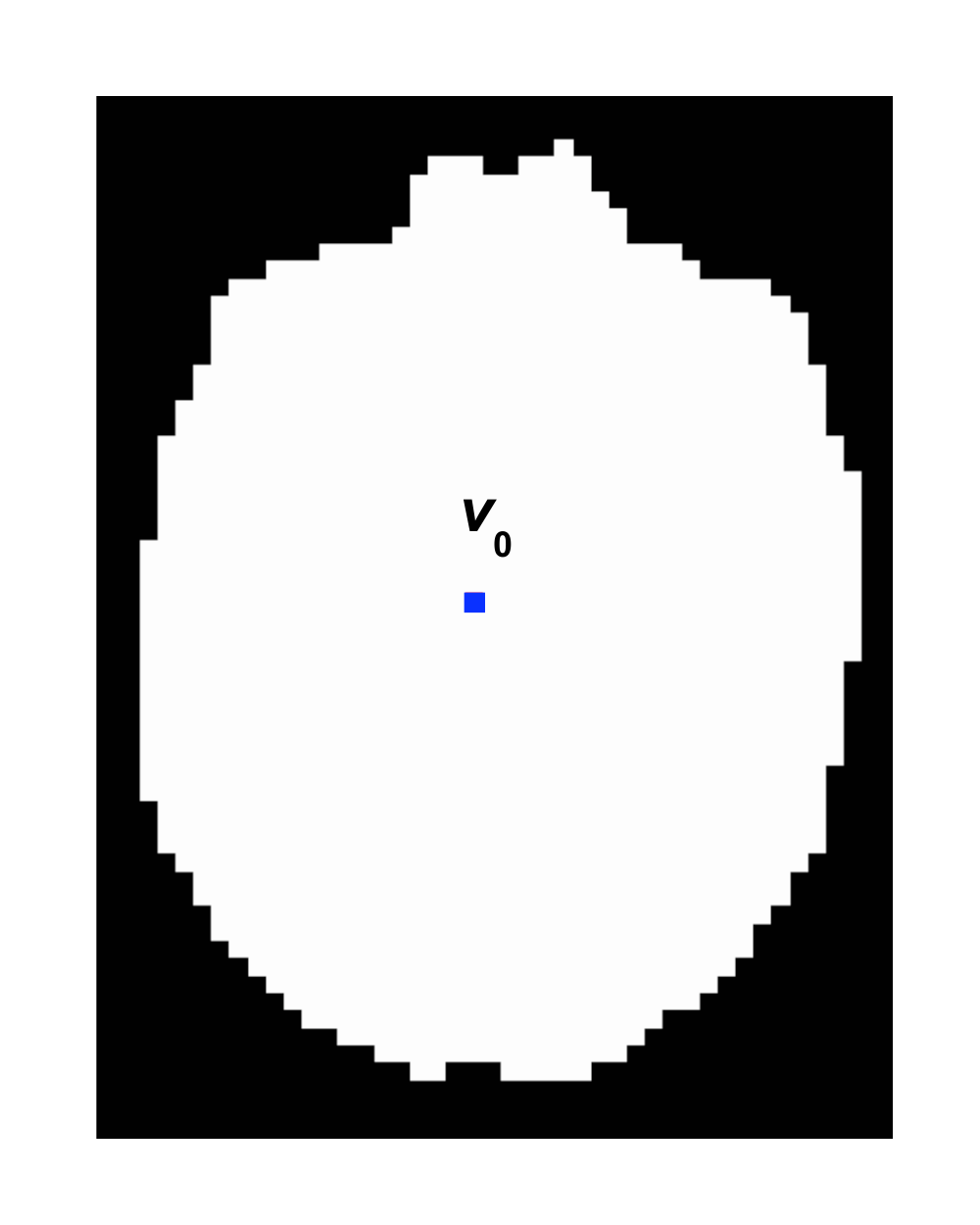}}
\caption{{\bf Single-voxel response}: (a) Scatter plot of simulated responses of voxel $v_0$ to conditions $A$ and $B$. A total of $300$ samples were drawn per condition (see text). (b) Spatial location of voxel $v_0$ (indicated by blue square) in the simulated single-slice voxel space.}
\label{fig:needle}
\end{center}
\end{figure*}

\subsection{The needle-in-the-haystack effect}\label{sec:needleinhay}

Suppose there exists some voxel in $\mathcal{V}$, say $v_0$, that exhibits a response difference to the experimental conditions such that the informativeness function identifies $v_0$ as being task-relevant, that is, $I(v_0) = 1$. Since $I(v_0) = 1$, by the SIN assumption it follows that each of the searchlights that contain $v_0$ should also be deemed to be informative as well.  Recall that, according to Theorem \ref{cor:corsym}, each voxel $v$ in $\mathcal{V}$ is contained in exactly $N_r$ searchlights where $N_r = |S_r(v)|$. It then follows that the signal-carrying voxel $v_0$ should be contained in $N_r$ searchlights, each being centered at a voxel in $S_r(v_0)$. Thus, a single signal-carrying voxel (a ``needle") should produce a cluster having $N_r$ voxels on the information map (a ``haystack"). 

To simulate this ``needle-in-the-haystack" effect, the task-relevant responses of $v_0$ in conditions $A$ and $B$ took the form illustrated in Figure \ref{fig:needle}(a). The responses to both conditions were drawn randomly from a normal distribution with standard deviation $\sigma = 1$. The voxel $v_0$'s mean response to condition $A$ was $\mu_A = +2$; and $\mu_B = -2$ for condition $B$. The responses of all other (non task-relevant) voxels were drawn from normal distributions having $\sigma=1$ where $\mu_A = \mu_B = 0$. To maximize the sensitivity of the searchlight statistic and emulate the requirements of the SIN assumption, a total of $300$ samples were drawn for each condition.  The spatial position of voxel $v_0$ is shown in blue in Figure \ref{fig:needle}(b). The voxel was placed far from the boundaries of the slice to avoid truncations of the searchlights and to emulate a uniform voxel-space in the vicinity of $v_0$. 

With this setup, the searchlight decomposition and testing procedure was implemented using the PyMVPA toolbox  \citep{Hanke:2009fk}. Each searchlight's informativeness was determined by evaluating the decodability of its responses, i.e., testing for the existence of a model that accurately classifies a sample's membership in each condition based on the searchlight's responses  \citep{Pereira:2011bl,Pereira:2009ba}. Decodability was tested using a linear Support Vector Machine (SVM) with a soft-margin regularization parameter, $C=1$. The searchlight statistic was the mean classification accuracy obtained using a Leave-One-Out (LOO) cross-validation procedure.

\begin{figure*}[!htbp]
\begin{center}
\subfigure[]{\includegraphics[width=0.58\textwidth]{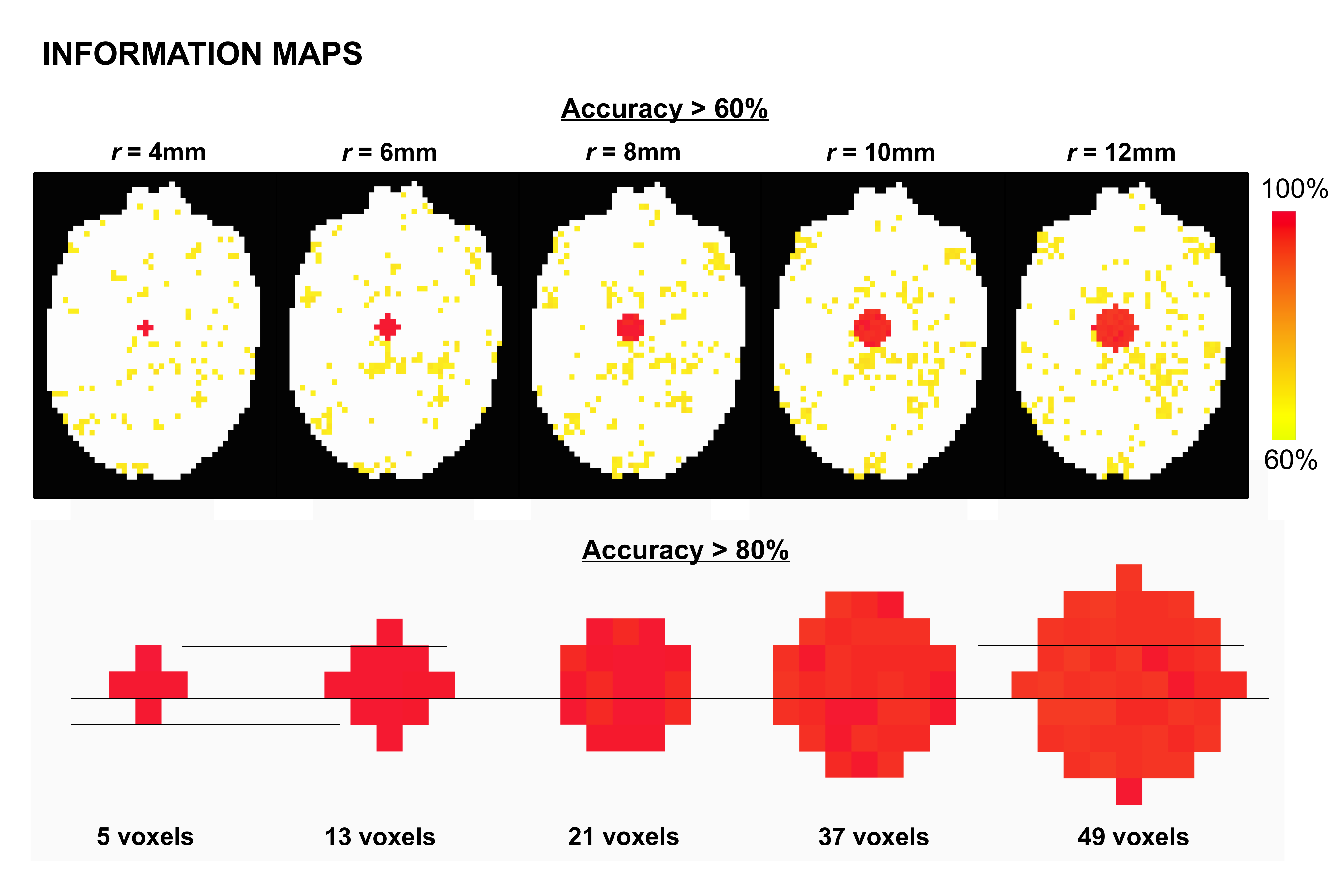}}
\subfigure[]{\includegraphics[width=0.41\textwidth]{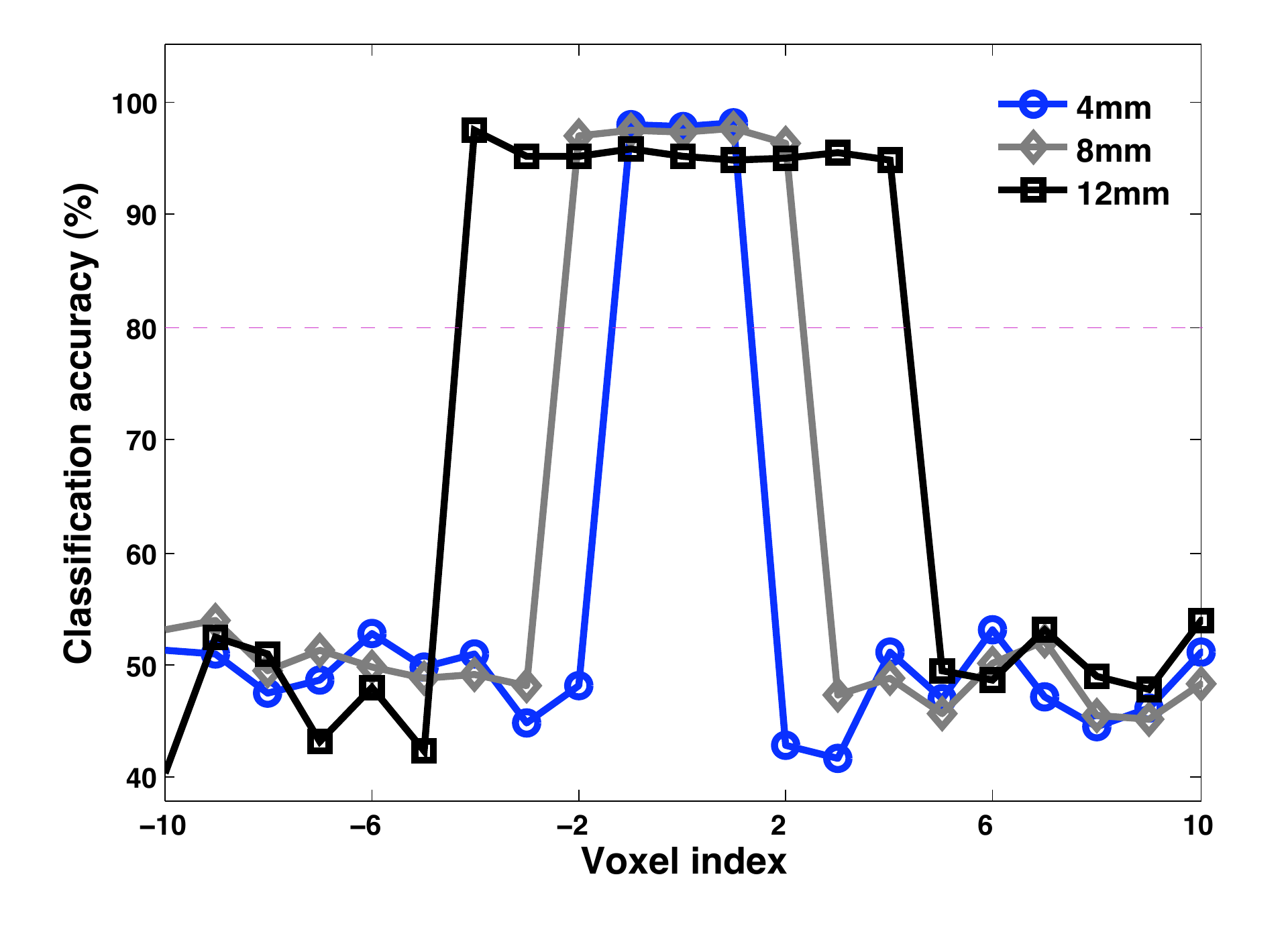}}
\caption{{{\bf Information maps for single-voxel response}: (a) Upper panel shows the information-maps across the entire slice as a function of increasing searchlight radius going from left to right (threshold = $60\%$ accuracy). Lower panel shows the corresponding expanded view of the high-accuracy clusters (threshold = $80\%$ accuracy) for each searchlight radius value. Horizontal lines are to provide a common reference to compare relative sizes of clusters. (b) Cross-sectional profile of the information map showing classification accuracies for searchlights centered at voxels on a horizontal 1D slice through the voxel $v_0$. The dotted horizontal line indicates the thresholding value of $80\%$. Only the profiles for radii $4$mm, $8$mm and $12$mm are shown.}}
\label{fig:needleped}
\end{center}
\end{figure*}

Figure \ref{fig:needleped}(a) shows the information maps obtained (thresholded at $60\%$). In the upper-panel, going from the left to the right in order of increasing radius, we see that there is a single high accuracy cluster (red-colored voxels) centered at the signal-carrying voxel $v_0$, and this cluster grows in size with increasing radius.  The lower-panel shows an expanded view of this high accuracy cluster, thresholded at $80\%$. Consistent with the predictions described above, for each radius, the size and shape of these clusters on the information map correspond exactly to the size, shape and location of the searchlight $S_r(v_0)$ centered at voxel $v_0$. Furthermore, consistent with Theorem \ref{thm:monotone},  the number of informative searchlights identified ($F_r$) increases in a monotonic manner with the radius of the searchlight, even though there is no difference in the actual information present or even any multivoxel response patterns to speak of.

Figure \ref{fig:needleped}(b) shows the values on the information map from a single 1D segment running horizontally through the voxel $v_0$ through the diameter of the searchlights centered at $v_0$. The voxel $v_0$ is assigned a value $0$. Consistent with the SIN assumption, the accuracies on the information map do not exhibit a smooth degradation as a function of the distance from $v_0$. Critically, this pedestal-like profile is unlike the profile that would be expected if the searchlights were the equivalent of a ``spatial smoothing" kernel on the information map. 

What is the comparable effect on the information map when the task-relevant signal is distributed over multiple voxels?
We next consider this scenario.

\subsection{The haystack-in-the-needle effect}

Suppose there are two voxels, $v_1$ and $v_2$ in $\mathcal{V}$, that conjointly exhibit a response difference to the experimental conditions. However, neither voxel by itself shows a task-relevant difference. That is, $I(\{v_1,v_2\}) = 1$ and $I(v_1) = I(v_2) = 0$. By the SIN assumption, every searchlight that contains both $v_1$ and $v_2$ should be informative, but searchlights that contain either $v_1$ or $v_2$ alone would not necessarily be informative. Recall that, according to Theorem \ref{cor:haystack}, a group of multiple voxels (i.e., having more than one voxel) is contained in strictly less than $N_r$ searchlights. It then follows that the signal-carrying voxel group $\{v_1,v_2\}$ should produce a cluster having \emph{less} than $N_r$ voxels on the information map, i.e., a multivoxel ``haystack''  should produce a ``needle"-like cluster, unlike the needle-in-the-haystack scenario in Section \ref{sec:needleinhay} above.

To simulate this ``haystack-in-the-needle" effect, the task-relevant responses in the two voxels $v_1$ and $v_2$ took the form shown in Figure \ref{fig:haystack}(a). The responses to each condition were drawn randomly from a normal distribution having standard deviation $\sigma = 1$. Each voxel's mean response to conditions $A$ and $B$ are shown as dotted lines. The voxel $v_1$ had an identical mean response to both conditions $A$ and $B$, specifically, $\mu_A = \mu_B = 0$ (the horizontal dotted line); while voxel $v_2$'s mean response to condition $A$ was $\mu_A = +0.5$ and to condition $B$ was $\mu_B = -0.5$ (indicated by each of the dotted vertical lines).  Importantly, the responses of voxel $v_1$ and $v_2$ to both conditions were correlated negatively. The response of voxel $v_1$ on condition $A$, denoted as $X_{1,A}$ was equal to $-X_{2,A}$, the response of voxel $v_2$ to condition $A$. Similarly, for condition $B$, $X_{1,B} = -X_{2,B}$. The simulated responses of all other voxels were drawn from distributions having $\sigma=1$ and $\mu_A = \mu_B = 0$, and were uncorrelated with the responses in either voxel $v_1$ or $v_2$. As with the previous simulation above, a total of $300$ samples were drawn for each condition. With signals of this form, the conjoint responses of voxels $v_1$ and $v_2$ to conditions $A$ and $B$ are linearly separable (see Figure \ref{fig:haystack}(a)). However, $A$ and $B$ cannot be distinguished from the responses in $v_1$, but should be weakly discriminable from the responses in $v_2$.

\begin{figure*}[htbp]
\begin{center}
\subfigure[]{\includegraphics[height=2.2in]{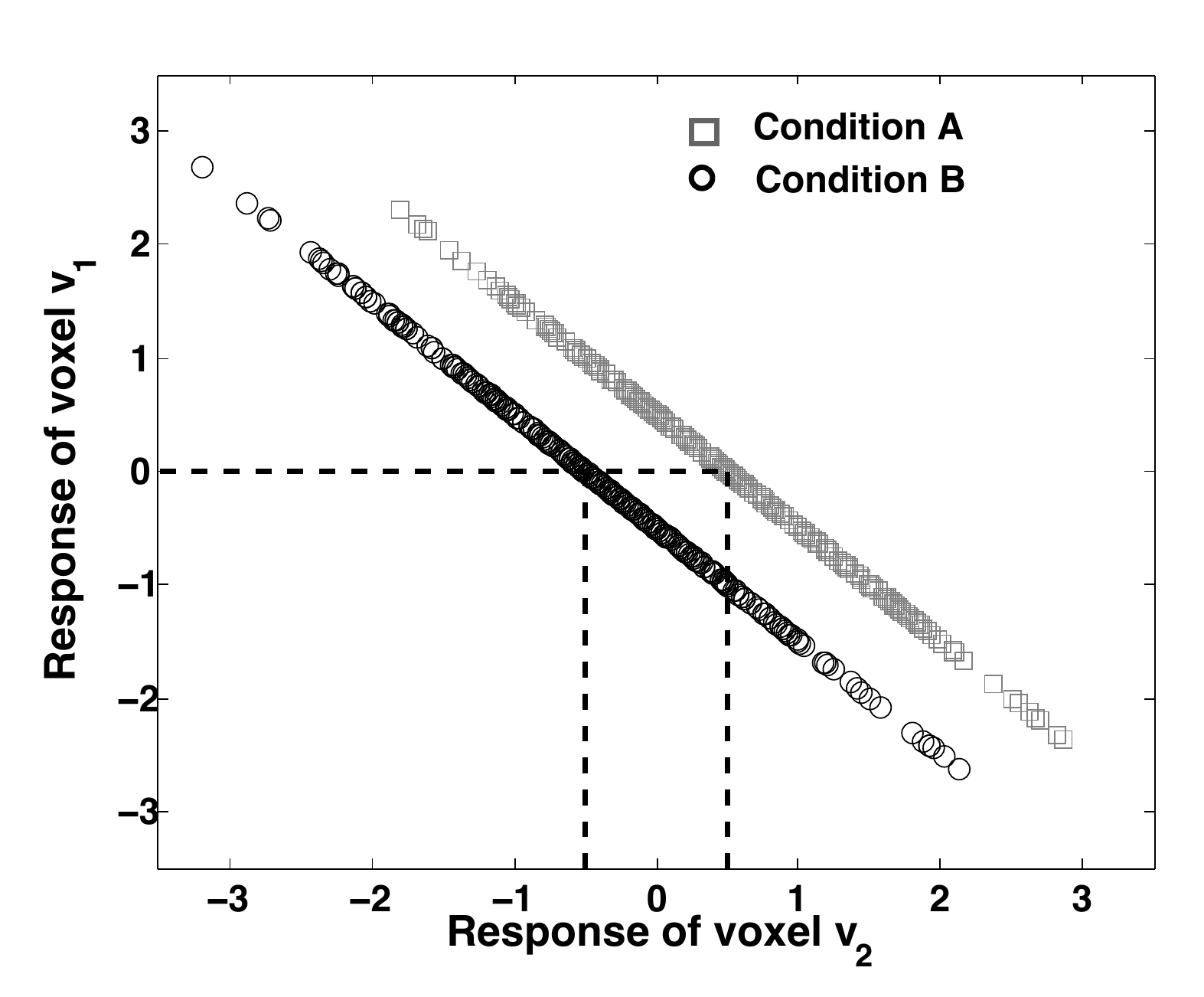}}
\hspace{0.5in}
\subfigure[]{\includegraphics[height=2.2in]{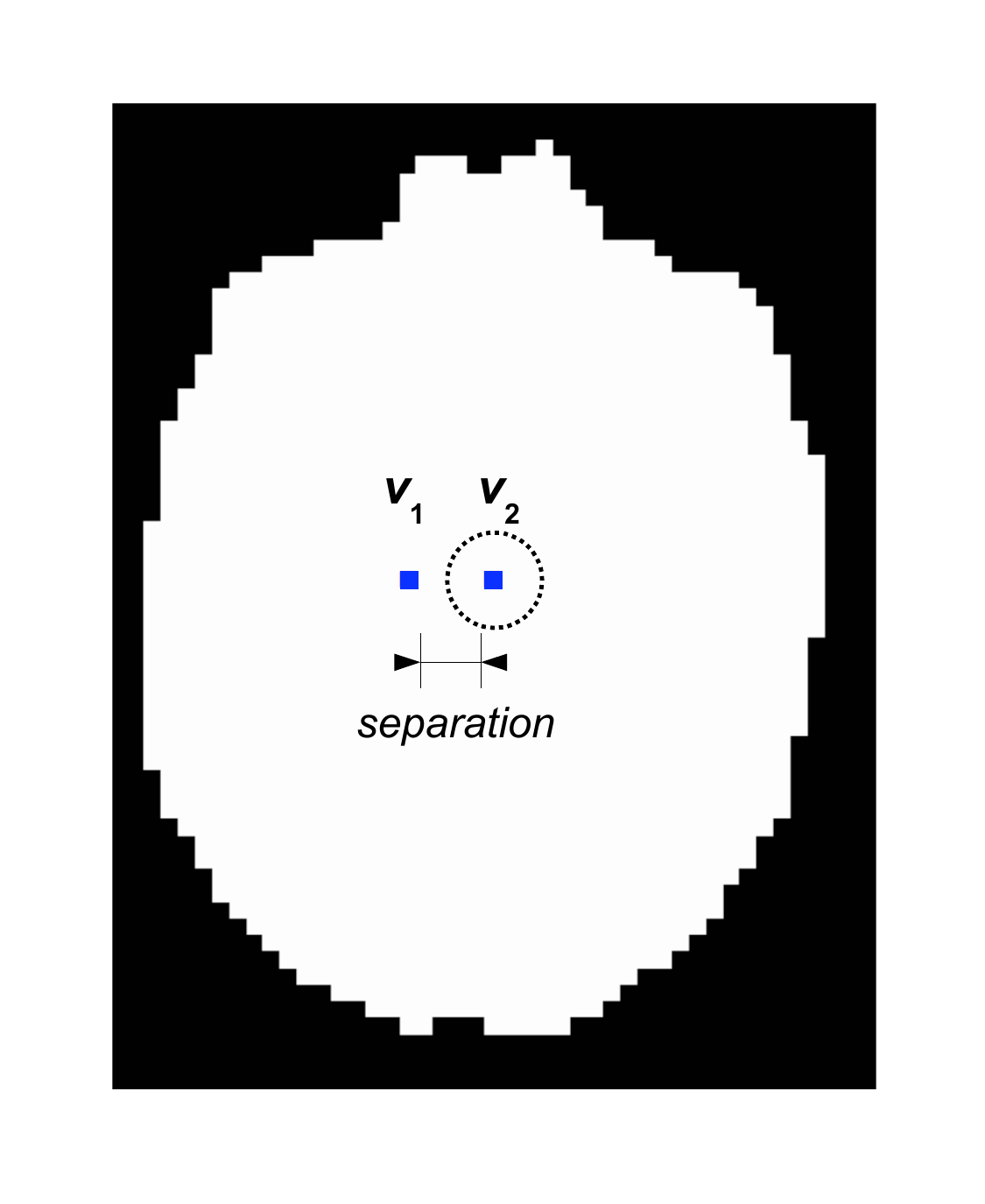}}
\caption{
{{\bf Multivoxel response}: (a) Scatter plot of simulated responses of voxels $v_1$ and $v_2$ to conditions $A$ (gray squares) and $B$ (black dots). Dotted lines indicate the mean response of each voxel alone to $A$ and $B$. (b) Spatial location of voxels $v_1$ and $v_2$ (indicated by blue squares) in the simulated single-slice voxel space. The separation between $v_1$ and $v_2$ was either 2 or 3 voxels (see text). The responses of voxel $v_2$ show a weak response difference to the conditions, indicated by the dotted circle.}}
\label{fig:haystack}
\end{center}
\end{figure*}

The relative spatial positions of $v_1$ and $v_2$, indicated as blue squares, are shown in Figure \ref{fig:haystack}(b). We considered two cases, where $v_1$ and $v_2$ were separated by $2$ voxels in one case; and by $3$ voxels in the other. When $v_1$ and $v_2$ have a separation of $2$ voxels, there is no one searchlight of radius $4$mm that can contain both of these voxels. With a separation of $3$ voxels, there are no searchlights of radius $4$ mm, $6$ mm, or $8$ mm that can contain both $v_1$ and $v_2$. With this setup, the searchlight decomposition and testing procedure was simulated in the same manner as in Section \ref{sec:needleinhay}.
 
Figure \ref{fig:haystackped} shows the portions of the information maps in the vicinity of voxels $v_1$ and $v_2$ (thresholded at $60\%$). In all the information maps, the above-threshold cluster takes the size and shape of the corresponding searchlight and is centered at voxel $v_2$, namely, the voxel exhibiting a weak response difference to conditions $A$ and $B$. This ``needle-in-the-haystack" organization is consistent with the simulations in Section \ref{sec:needleinhay}, and is invariant to the number of voxels separating $v_1$ and $v_2$. 

Now, observe that the clusters in several, but not all, of the information maps contain sub-clusters consisting of voxels  having high classification accuracies (indicated in red). These voxels on the information map correspond to the centers of searchlights that contain both $v_1$ and $v_2$. As required by Theorem \ref{cor:haystack}, for each radius, the number of high-accuracy voxels in the cluster are less than $N_r$. Due to the geometric constraint defined by the separation between $v_1$ and $v_2$, the presence of any high-accuracy voxels at all in an information map depends on the radius of the searchlights used. For example, information maps obtained with searchlights of radius $4$ mm do not contain any high-accuracy voxels for both separations (top row), while the information maps for searchlights of radius $8$mm contain high-accuracy voxels for the $2$ voxel separation but not for the $3$ voxel separation. 

Figures \ref{fig:haystackped}(a) and (b) show the 1D cross-section of the information map through the horizontal diameter of the clusters for the $2$ voxel and $3$ voxel separations respectively. As evident, there is a ``smearing", rather than smoothing, of the accuracies with growing radius values, as in Figure \ref{fig:needleped}(b). Furthermore, when a searchlight is large enough to include both $v_1$ and $v_2$, there is a large increase in the classification accuracy. 

\begin{figure*}[htbp]
\begin{center}
\subfigure[]{\includegraphics[height=4in]{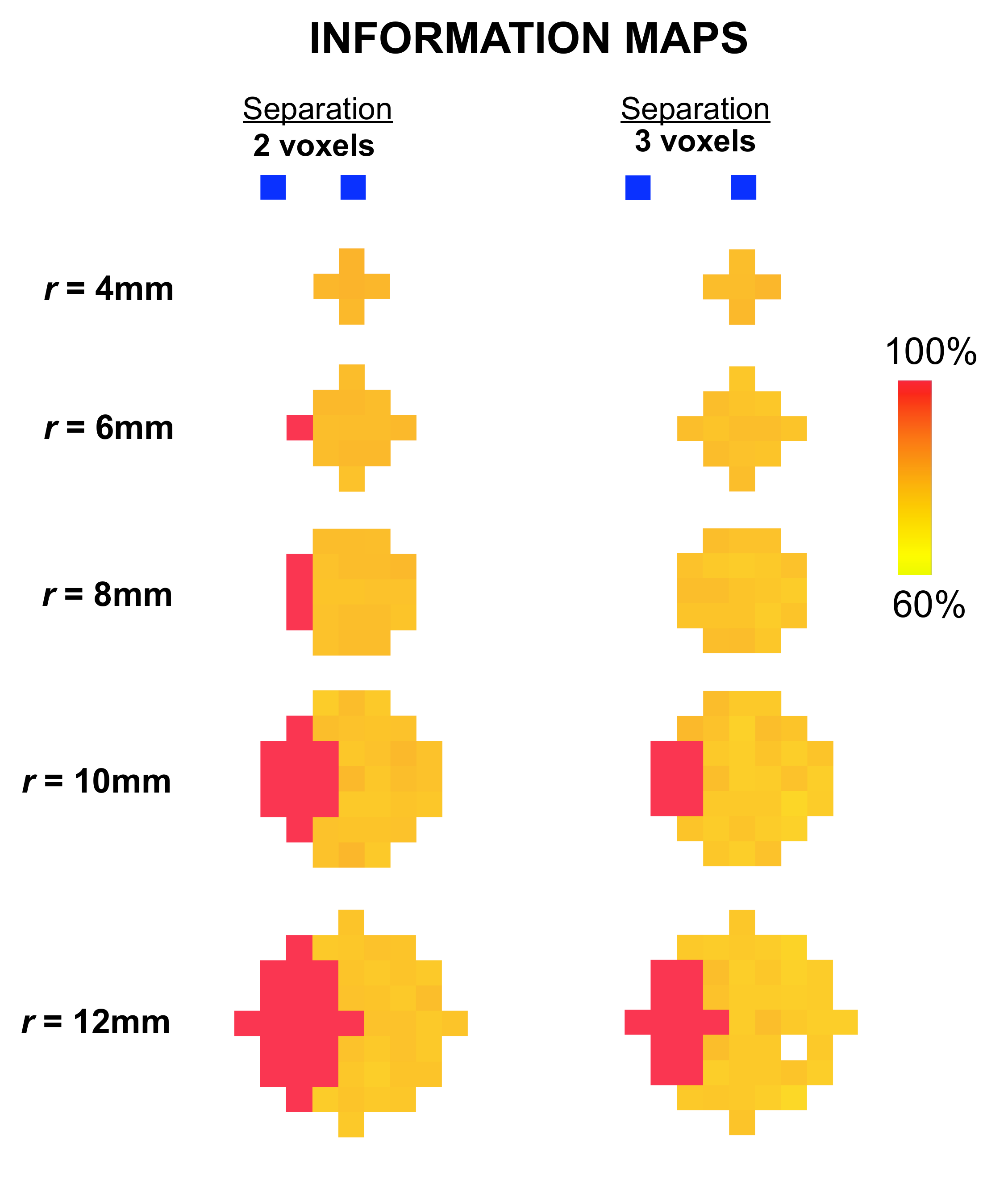}}\\
\subfigure[]{\includegraphics[width=0.49\textwidth]{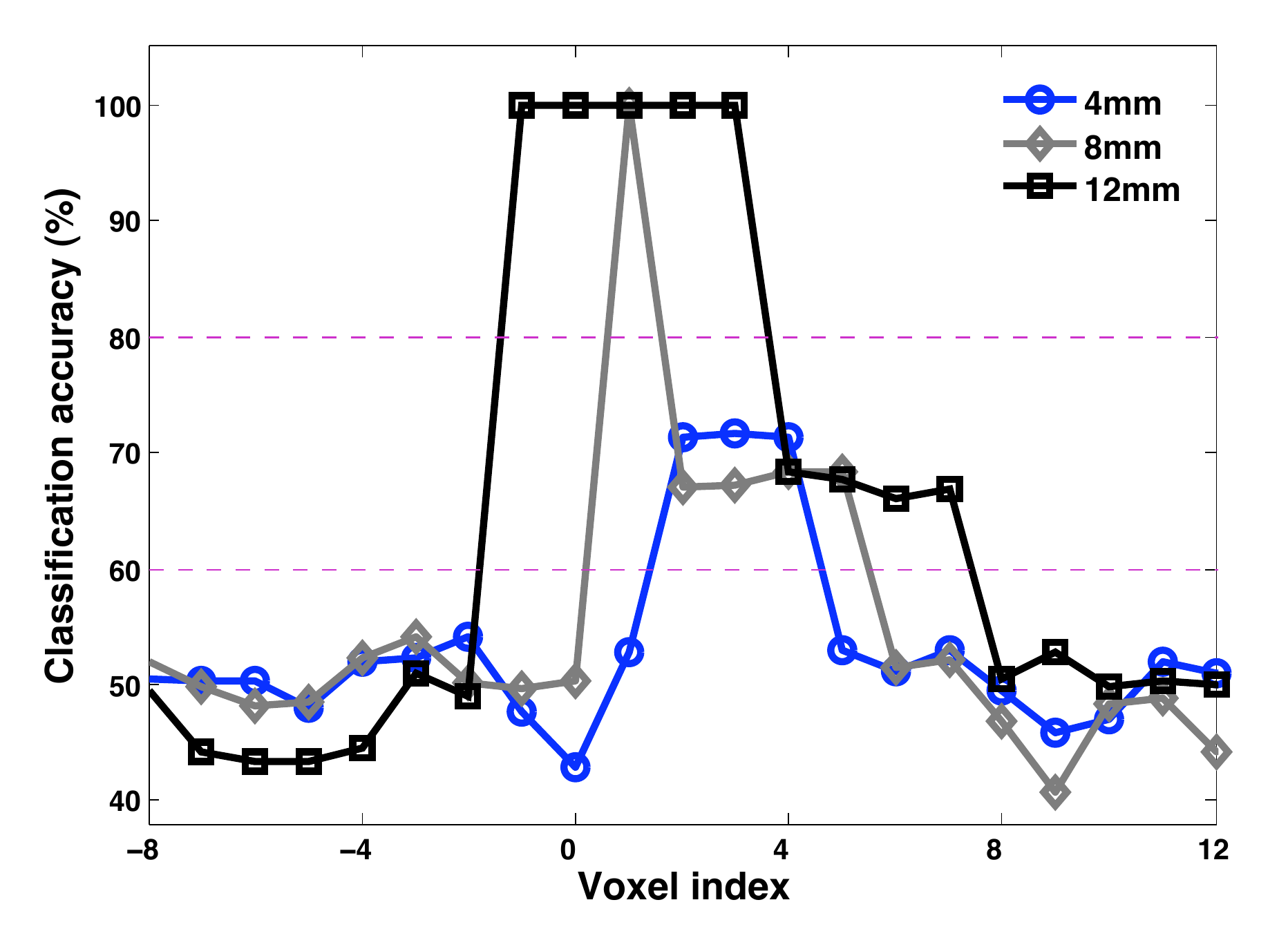}}
\subfigure[]{\includegraphics[width=0.49\textwidth]{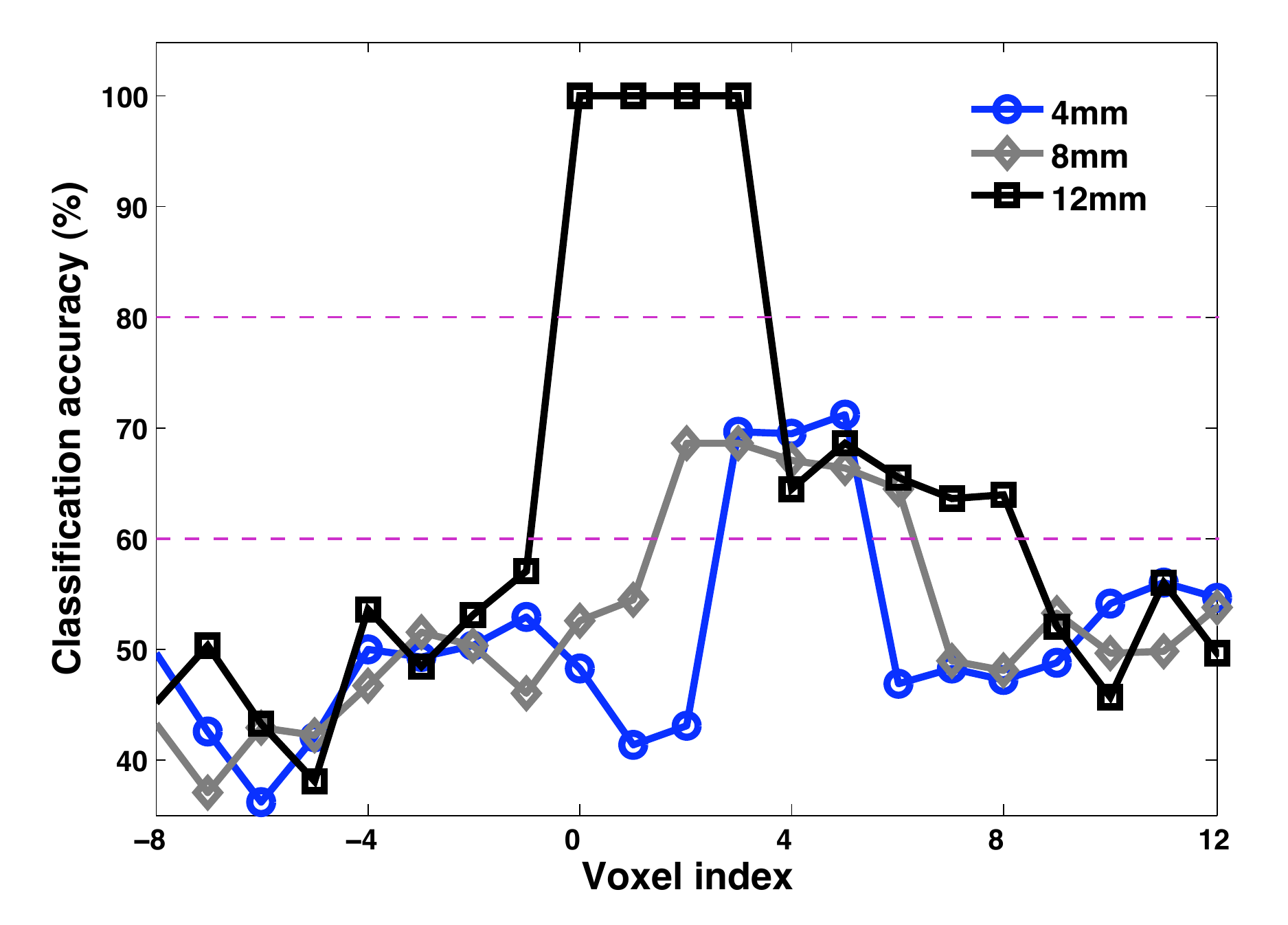}}
\caption{{{\bf Information maps for multivoxel response}: (a) Left and right panels show portions of information map (threshold = $60\%$ accuracy) centered around voxels $v_1$ and $v_2$ as a function of increasing searchlight radius (top to bottom); and the separation between $v_1$ and $v_2$ -- 2 voxel separation (left panel), and 3 voxel separation (right panel). Plots (b) and (c):  Cross-sectional profile of the information map showing classification accuracies for searchlights centered at voxels on a horizontal 1D slice through the voxels $v_1$ and $v_2$ for separation of 2 voxels (b), and 3 voxels (c). The lower dotted horizontal line indicates the thresholding value of $60\%$, and the upper dotted line indicates the regime of the high-accuracy searchlights ($80\%$).}}
\label{fig:haystackped}
\end{center}
\end{figure*}

The above simulations confirm the basic statistical premise motivating the searchlight-procedure, namely, the ability of a multivoxel pattern analysis method to detect distributed response patterns. However, for any radius, the size of the clusters produced by multivoxel response patterns are \underline{smaller} than those produced by single voxel response-differences.  Consistent with Theorem \ref{thm:monotone}, the number of informative searchlights identified increases in a monotonic manner with the radius of the searchlight.

\subsection{Whole-brain inflation maps}

The previous two simulations demonstrated signal-dependent effects caused by the sampling bias inherent in the searchlight decomposition. However, according to Theorem \ref{thm:monotone}, there should be a monotonic increase in the number of informative searchlights as a function of radius, irrespective of the actual distribution of task-relevant voxels/voxel-groups across the brain. This monotonic scaling of the size of the ``blobs" on the information map makes plausible a rather unusual scenario -- an information map where \emph{every} searchlight in the brain is deemed to be informative. 

This scenario was motivated by results recently reported by \citet{Poldrack:2009fm}. In that study, information maps were generated using searchlights of radius $4$ mm and $8$ mm. Rather remarkably, with a radius of $8$ mm, only one region in the information map (the bilateral dorsolateral prefrontal cortex) was found to be \emph{uninformative} while every other searchlight was informative. This whole-brain coverage was, however, not the case with the $4$ mm searchlights. Given the inflationary relationship between $F_r$ (the number of informative searchlights) and searchlight radius $r$ that established in the previous sections, curiosity asked: could an informative whole-brain arise (i.e., $F_r = |\mathcal{V}|$) by random chance with a suitably chosen searchlight radius?

This question can be formulated as a covering problem. Consider a finite 3D voxel space corresponding to one containing the brain,  approximated as a cubic volume of size $N_X \times N_Y \times N_Z$, where $N_i$ is the number of voxels along the principal direction $i$. Suppose there is a minimum covering set of searchlights $C_r \subset \mathcal{S}_r(\mathcal{V})$ such that every voxel in $\mathcal{V}$ is contained in some searchlight in $C_r$. Recall that a single-voxel signal can produce a cluster having $N_r$ voxels on the information map, due to Theorem \ref{cor:corsym}. If the central voxel of each of the searchlights in $C_r$ was informative, it would follow that searchlights centered at every voxel in every one of the searchlights in $C_r$ would also be informative. Since every voxel in $\mathcal{V}$ is present in some searchlight in $C_r$, it implies that a rather sparse distribution of informative single-voxels specified by $C_r$ could produce an information map where \emph{every} searchlight in $\mathcal{S}_r(\mathcal{V})$ would be informative (with the proviso that the SIN assumption holds true.)

\begin{figure}[htbp]
\begin{center}
\includegraphics[width=0.5 \textwidth]{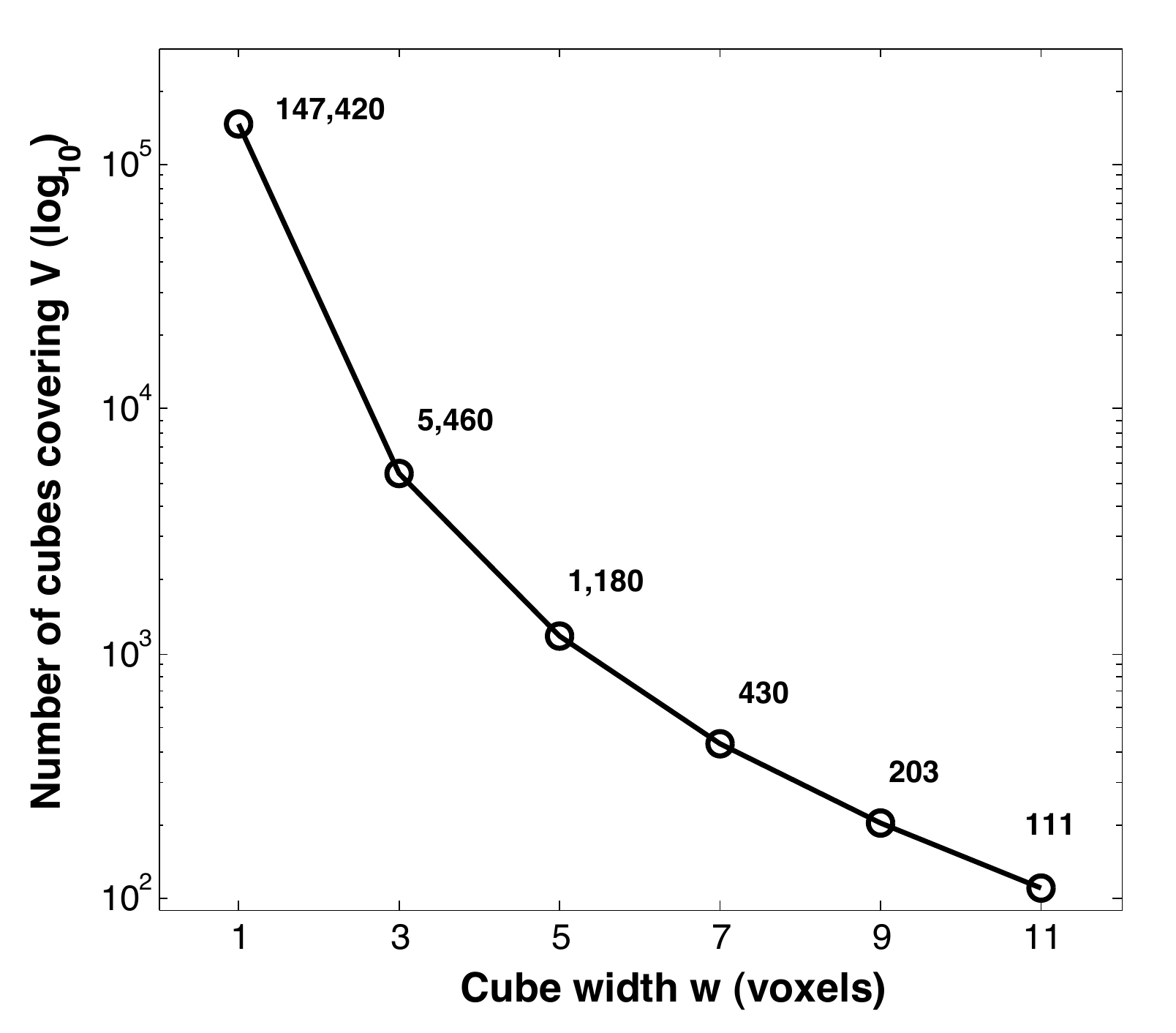}
\caption{{Number of non-intersecting cubes required to fill a volume of dimensions $N_x = 52, N_y = 63, N_z = 45$ ($N_i$ = number of voxels in principal direction $i$) as a function of cube-size. Values next to each data-point indicate the actual number of filling cubes for each cube-size.}}
\label{fig:fillingvol}
\end{center}
\end{figure}

The sparsity of these informative single-voxels can be readily approximated if we use cubical volumes as a proxy for the spherical shape of the searchlight volumes. A cube of side $w$ voxels would fully contain a sphere having radius $r = w/2$, and would be fully contained in a sphere of radius $w\sqrt3/2$. With this simplification, the minimal number of searchlight cubes required to cover the voxel space $\mathcal{V}$ is readily approximated as the volume of the voxel space divided by the volume of each searchlight cube, that is, $\approx \lceil (N_X N_Y N_Z)/(w^3) \rceil$. 

For voxels of size $3\times3\times3$ mm$^3$, we approximate the size of the voxel space with the following values $N_X = 52$ voxels, $N_Y = 63$ voxels and $N_Z = 45$ voxels. Figure \ref{fig:fillingvol} shows the minimum number of cubical volumes required to cover $\mathcal{V}$ as a function of $w$, where $w$ took values $1, 3, 5, 7, 9, 11$. 

A searchlight cube of side $w=1$ is equivalent to a single voxel so the size of the covering set $C_r$ is equal to the total number of voxels in $\mathcal{V}$, namely, $147420$. However, increasing values of $w$ produce a rapid decrease in the size of the covering set. For $w = 3$ voxels, a cubical volume that would fully contain a spherical searchlight of radius $4$ mm, a total of $5460$ equally spaced signal-carrying voxels can produce an information map where every searchlight is informative. However, for a cubical volume with side $w=7$ voxels, corresponding to spherical volumes of radius $8$mm, a mere $430$ voxels are required for such a fully informative map. Stated differently, an information map with a single task-relevant cluster made up of every voxel in $\mathcal{V}$ can be generated from a mere $430$ regularly spaced voxels of the $147,420$ voxels in $\mathcal{V}$, that is, $430$ voxels that enable the conditions to be distinguished whether due to the presence of true signal or by random chance. This potential for a small number of single voxels (i.e., $\approx 0.003\%$ of $|\mathcal{V}|$) to drive the structure of the entire information map simply by the choice of the searchlight radius presents an important consideration for drawing neurobiological interpretation.

\section{Discussion}

Knowledge of the actual information-carrying voxels in each informative searchlight would make the information map irrelevant. These actual informative voxels could be directly reported, hence resolving the over-counting that arises from their inclusion in multiple searchlights. One possible implementation would be to identify task-relevant voxels in each searchlight, and then combine these identified voxels across searchlights. However, requiring the identification of the actual informative voxels in each searchlight could reduce the generality of the searchlight method. When pattern classifiers are used to compute the searchlight statistic, each voxel (or feature) in a searchlight is typically assigned a weight, and the weighted combination of the multivoxel responses is used to make a classification decision. However, the specific basis for assigning weights to individual features is highly dependent on the specific machine learning algorithm and its inductive assumptions \citep{mitchell1980need,wolpert1996lack,guyon2002gene,Pereira:2009ba}. Consequently, appropriate techniques would be required to allow results to be compared across studies that use different MVPA-techniques. 

Until such advances are made, the analytical framework described above provides several constraints on alternate interpretations of the information map. Our results present a strong argument against measuring the sensitivity of information mapping by a count of the number of informative searchlights. The seemingly high sensitivity of the searchlight method as judged by such a performance measure in part has a rather trivial explanation. Specifically, an explanation in the obligatory geometric properties of the searchlight-method as discussed above rather than an explanation related to underlying neural organization, or the sophisticated machine-learning algorithms used to analyze multivoxel response patterns, or the widely discussed merits of multivariate statistical evaluations. Indeed, the upshot of the optimistic scaling of this performance measure is that it is maximal when explicitly assuming a highly sensitive and robust MVPA technique, namely one satisfying the superset informativeness (SIN) assumption.

\section{Acknowledgments}
This work was supported by the U.S. Army Research Office through the Institute for Collaborative Biotechnologies under Contract No. W911NF-09-D-0001.

\section{References}

\begin{thebibliography}{35}
\expandafter\ifx\csname natexlab\endcsname\relax\def\natexlab#1{#1}\fi
\expandafter\ifx\csname url\endcsname\relax
  \def\url#1{\texttt{#1}}\fi
\expandafter\ifx\csname urlprefix\endcsname\relax\def\urlprefix{URL }\fi

\bibitem[{Alink et~al.(2011)Alink, Euler, Kriegeskorte, Singer, and
  Kohler}]{Alink:2011sf}
Alink, A., Euler, F., Kriegeskorte, N., Singer, W., Kohler, A., 2011. Auditory
  motion direction encoding in auditory cortex and high-level visual cortex.
  Human Brain Mapping.

\bibitem[{Chadwick et~al.(2010)Chadwick, Hassabis, Weiskopf, and
  Maguire}]{Chadwick:2010uq}
Chadwick, M.~J., Hassabis, D., Weiskopf, N., Maguire, E.~A., 2010. Decoding
  individual episodic memory traces in the human hippocampus. Current Biology
  20~(6), 544--7.

\bibitem[{Chen et~al.(2011)Chen, Namburi, Elliott, Heinzle, Soon, Chee, and
  Haynes}]{Chen:2011vl}
Chen, Y., Namburi, P., Elliott, L.~T., Heinzle, J., Soon, C.~S., Chee, M.
  W.~L., Haynes, J.-D., 2011. Cortical surface-based searchlight decoding.
  Neuroimage 56~(2), 582--92.

\bibitem[{Connolly et~al.(2012)Connolly, Guntupalli, Gors, Hanke, Halchenko,
  Wu, Abdi, and Haxby}]{Connolly:2012gl}
Connolly, A.~C., Guntupalli, J.~S., Gors, J., Hanke, M., Halchenko, Y.~O., Wu,
  Y.-C., Abdi, H., Haxby, J.~V., 2012. The representation of biological classes
  in the human brain. Journal of Neuroscience 32~(8), 2608--18.

\bibitem[{Cox and Savoy(2003)}]{Cox:2003vc}
Cox, D.~D., Savoy, R.~L., Jun 2003. Functional magnetic resonance imaging
  (fmri) "brain reading": detecting and classifying distributed patterns of
  fmri activity in human visual cortex. Neuroimage 19~(2 Pt 1), 261--70.

\bibitem[{Formisano and Kriegeskorte(2012)}]{Formisano:2012fk}
Formisano, E., Kriegeskorte, N., 2012. Seeing patterns through the hemodynamic
  veil - the future of pattern-information f{MRI}. Neuroimage 62~(2), 1249--56.

\bibitem[{Golomb and Kanwisher(2011)}]{Golomb:2011ph}
Golomb, J.~D., Kanwisher, N., 2011. Higher level visual cortex represents
  retinotopic, not spatiotopic, object location. Cerebral Cortex Epub.

\bibitem[{Guyon et~al.(2002)Guyon, Weston, Barnhill, and
  Vapnik}]{guyon2002gene}
Guyon, I., Weston, J., Barnhill, S., Vapnik, V., 2002. Gene selection for
  cancer classification using support vector machines. Machine learning 46~(1),
  389--422.

\bibitem[{Hanke et~al.(2009)Hanke, Halchenko, Sederberg, Hanson, Haxby, and
  Pollmann}]{Hanke:2009fk}
Hanke, M., Halchenko, Y.~O., Sederberg, P.~B., Hanson, S.~J., Haxby, J.~V.,
  Pollmann, S., 2009. Py{MVPA}: A python toolbox for multivariate pattern
  analysis of f{MRI} data. Neuroinformatics 7~(1), 37--53.

\bibitem[{Haxby et~al.(2001)Haxby, Gobbini, Furey, Ishai, Schouten, and
  Pietrini}]{Haxby:2001kl}
Haxby, J.~V., Gobbini, M.~I., Furey, M.~L., Ishai, A., Schouten, J.~L.,
  Pietrini, P., 2001. Distributed and overlapping representations of faces and
  objects in ventral temporal cortex. Science 293~(5539), 2425--30.

\bibitem[{Haynes(2006)}]{Haynes:2006vv}
Haynes, J., 2006. {Decoding mental states from brain activity in humans}.
  Nature Reviews Neuroscience 7~(7), 523--534.

\bibitem[{Haynes et~al.(2007)Haynes, Sakai, Rees, Gilbert, Frith, and
  Passingham}]{Haynes:2007dy}
Haynes, J.-D., Sakai, K., Rees, G., Gilbert, S., Frith, C., Passingham, R.~E.,
  Feb. 2007. {Reading Hidden Intentions in the Human Brain}. Current Biology
  17~(4), 323--328.

\bibitem[{Jimura and Poldrack(2012)}]{Jimura:2012cs}
Jimura, K., Poldrack, R.~A., 2012. Analyses of regional-average activation and
  multivoxel pattern information tell complementary stories. Neuropsychologia
  50~(4), 544--52.

\bibitem[{Johnson et~al.(2009)Johnson, McDuff, Rugg, and
  Norman}]{Johnson:2009kx}
Johnson, J.~D., McDuff, S. G.~R., Rugg, M.~D., Norman, K.~A., 2009.
  Recollection, familiarity, and cortical reinstatement: a multivoxel pattern
  analysis. Neuron 63~(5), 697--708.

\bibitem[{Kaplan and Meyer(2012)}]{Kaplan:2012gd}
Kaplan, J.~T., Meyer, K., 2012. Multivariate pattern analysis reveals common
  neural patterns across individuals during touch observation. Neuroimage
  60~(1), 204--12.

\bibitem[{Kriegeskorte et~al.(2006)Kriegeskorte, Goebel, and
  Bandettini}]{Kriegeskorte:2006bf}
Kriegeskorte, N., Goebel, R., Bandettini, P., 2006. Information-based
  functional brain mapping. Proceedings of the National Academy of Sciences
  103~(10), 3863--8.

\bibitem[{Mitchell(1980)}]{mitchell1980need}
Mitchell, T., 1980. The need for biases in learning generalizations. Tech. Rep.
  CBM-TR-5-110, Department of Computer Science, Rutgers University.

\bibitem[{Morgan et~al.(2011)Morgan, Macevoy, Aguirre, and
  Epstein}]{Morgan:2011ys}
Morgan, L.~K., Macevoy, S.~P., Aguirre, G.~K., Epstein, R.~A., 2011. Distances
  between real-world locations are represented in the human hippocampus.
  Journal of Neuroscience 31~(4), 1238--45.

\bibitem[{Mur et~al.(2008)Mur, Bandettini, and Kriegeskorte}]{Mur:2008cn}
Mur, M., Bandettini, P.~A., Kriegeskorte, N., 2008. {Revealing representational
  content with pattern-information f{MRI}--an introductory guide}. Social
  Cognitive and Affective Neuroscience 4~(1), 101--109.

\bibitem[{Nestor et~al.(2011)Nestor, Plaut, and Behrmann}]{Nestor:2011pf}
Nestor, A., Plaut, D.~C., Behrmann, M., 2011. Unraveling the distributed neural
  code of facial identity through spatiotemporal pattern analysis. Proceedings
  of the National Academy of Sciences 108~(24), 9998--10003.

\bibitem[{Norman et~al.(2006)Norman, Polyn, Detre, and Haxby}]{Norman:2006gy}
Norman, K.~A., Polyn, S.~M., Detre, G.~J., Haxby, J.~V., Sep. 2006. {Beyond
  mind-reading: multi-voxel pattern analysis of fMRI data}. Trends in Cognitive
  Sciences 10~(9), 424--430.

\bibitem[{Oosterhof et~al.(2012)Oosterhof, Tipper, and
  Downing}]{Oosterhof:2012vb}
Oosterhof, N.~N., Tipper, S.~P., Downing, P.~E., Jan 2012. Viewpoint
  (in)dependence of action representations: An {MVPA} study. J Cogn Neurosci.

\bibitem[{Oosterhof et~al.(2011)Oosterhof, Wiestler, Downing, and
  Diedrichsen}]{Oosterhof:2011ad}
Oosterhof, N.~N., Wiestler, T., Downing, P.~E., Diedrichsen, J., May 2011. A
  comparison of volume-based and surface-based multi-voxel pattern analysis.
  Neuroimage 56~(2), 593--600.

\bibitem[{Oosterhof et~al.(2010)Oosterhof, Wiggett, Diedrichsen, Tipper, and
  Downing}]{Oosterhof:2010jq}
Oosterhof, N.~N., Wiggett, A.~J., Diedrichsen, J., Tipper, S.~P., Downing,
  P.~E., Aug 2010. Surface-based information mapping reveals crossmodal
  vision-action representations in human parietal and occipitotemporal cortex.
  Journal of Neurophysiology 104~(2), 1077--89.

\bibitem[{Peelen and Kastner(2011)}]{Peelen:2011jt}
Peelen, M.~V., Kastner, S., Jul 2011. A neural basis for real-world visual
  search in human occipitotemporal cortex. Proceedings of the National Academy
  of Sciences 108~(29), 12125--30.

\bibitem[{Pereira and Botvinick(2011)}]{Pereira:2011bl}
Pereira, F., Botvinick, M., May 2011. {Information mapping with pattern
  classifiers: A comparative study}. NeuroImage 56~(2), 476--496.

\bibitem[{Pereira et~al.(2009)Pereira, Mitchell, and
  Botvinick}]{Pereira:2009ba}
Pereira, F., Mitchell, T., Botvinick, M., Mar. 2009. {Machine learning
  classifiers and fMRI: A tutorial overview}. NeuroImage 45~(1), S199--S209.

\bibitem[{Poldrack et~al.(2009)Poldrack, Halchenko, and
  Hanson}]{Poldrack:2009fm}
Poldrack, R.~A., Halchenko, Y.~O., Hanson, S.~J., Nov. 2009. {Decoding the
  Large-Scale Structure of Brain Function by Classifying Mental States Across
  Individuals}. Psychological Science 20~(11), 1364--1372.

\bibitem[{Serences and Saproo(2012)}]{Serences:2012fk}
Serences, J.~T., Saproo, S., Mar 2012. Computational advances towards linking
  bold and behavior. Neuropsychologia 50~(4), 435--46.

\bibitem[{Soon et~al.(2008)Soon, Brass, Heinze, and Haynes}]{Soon:2008bl}
Soon, C.~S., Brass, M., Heinze, H.-J., Haynes, J.-D., Apr. 2008. {Unconscious
  determinants of free decisions in the human brain}. Nature Neuroscience
  11~(5), 543--545.

\bibitem[{Stokes et~al.(2011)Stokes, Saraiva, Rohenkohl, and
  Nobre}]{Stokes:2011ly}
Stokes, M., Saraiva, A., Rohenkohl, G., Nobre, A.~C., Jun 2011. Imagery for
  shapes activates position-invariant representations in human visual cortex.
  Neuroimage 56~(3), 1540--5.

\bibitem[{Tong(2010)}]{Tong2010}
Tong, F., Dec. 2010. {Pattern Classification Analysis}. Annual Review of
  Psychology 63~(1), 110301102248092.

\bibitem[{Wagner and Rissman(2010)}]{Wagner2010}
Wagner, A.~D., Rissman, J., Dec. 2010. {Distributed representations in memory:
  insights from functional brain imaging}. Annual Review of Psychology 63~(1),
  110301102248092.

\bibitem[{Wolpert(1996)}]{wolpert1996lack}
Wolpert, D., 1996. The lack of a priori distinctions between learning
  algorithms. Neural computation 8~(7), 1341--1390.

\bibitem[{Woolgar et~al.(2011)Woolgar, Thompson, Bor, and
  Duncan}]{Woolgar:2011kx}
Woolgar, A., Thompson, R., Bor, D., Duncan, J., May 2011. Multi-voxel coding of
  stimuli, rules, and responses in human frontoparietal cortex. Neuroimage
  56~(2), 744--52.

\end{thebibliography}

\end{document}